\documentclass[12pt]{article}
\usepackage{amsmath}
\usepackage{sw20elba}
\usepackage{amsfonts}
\usepackage{amscd}
\usepackage[T1]{fontenc}
\usepackage[utf8]{inputenc}

\newtheorem{theorem}{Theorem}

\newtheorem{definition}[theorem]{Definition}

\newtheorem{lemma}[theorem]{Lemma}

\newtheorem{proposition}[theorem]{Proposition}

\newenvironment{proof}[1][Proof]{\noindent\textbf{#1.} }{\ \rule{0.5em}{0.5em}}
\begin{document}

\title{DE DONDER FORM FOR SECOND ORDER GRAVITY}
\author{J\k{e}drzej \'{S}niatycki\thanks{%
Department of Mathematics and Statistics, University of Calgary and
Department of Mathematics and Statistics, University of Victoria. email:
sniatyck@ucalgary.ca} \ and O\u{g}ul Esen\thanks{%
Department of Mathematics, Gebze Technical University, 41400 Gebze, Kocaeli,
Turkey, email: oesen@gtu.edu.tr}}
\date{}
\maketitle

\begin{abstract}
We show that the De Donder form for second order gravity, defined in terms
of Ostrogradski's version of the Legendre transformation applied to all
independent variables, is globally defined by its local coordinate
descriptions. It is a natural differential operator applied to the
diffeomorphism invariant Lagrangian of the theory.
\end{abstract}

\section{Introduction}

In 1929, De Donder formulated an approach to study first order variational
problems for several independent variables in terms of a differential form
obtained by the Legendre transformation in each independent variable \cite%
{dedonder29} and \cite{dedonder35}. The De Donder form is a field theory
analogue of the Poincar\'{e}-Cartan form, which was introduced for a single
independent variable. It is a basis of the multisymplectic formulation of
field theory, which is called also a polysymplectic theory or De Donder-Weyl
theory. The first application of the De Donder form to general relativity in
the Palatini formulation \cite{palatini} was given in \cite{js1970}. For
further developments see \cite{bsf}, \cite{gotay}, \cite{kanatchikov}, \cite%
{szczyrba} and references quoted there.\smallskip

In 1936, Lepage \cite{lepage} constructed a family of forms, each of which
can be used in the same way as the De Donder form to reduce the original
variational problem to a system of equations in exterior differential forms.
In 1977, Aldaya and Azc\'{a}rraga \cite{AA} studied generalizations of the
Lepage construction to higher order Lagrangians, for which they used the
term Poincar\'{e}-Cartan forms. Here, we use the term De Donder form for the
Poincar\'{e}-Cartan form of Aldaya and Azc\'{a}rraga which is obtained from
the Lagrangian by Ostrogradski's generalization of the Legendre
transformation \cite{ostrogradski} in all independent variables.\smallskip\ 

The usual expression for a De Donder form is given in terms of coordinates
on an appropriate jet bundle induced by a coordinate patch on the space of
variables. For a generic Lagrangian, if the number of independent and the
number of dependent variables are greater than 1, this expression depends on
the choice of coordinates. Therefore, it does not define a global form. This
leads to a search for additional geometric structures, which would ensure
global existence of such forms, see \cite{campos}, \cite{de leon - rodrigues}
and references cited there. \smallskip

The aim of this paper is to show that, for second order general relativity
with a diffeomorphism invariant Lagrangian $L$, the coordinate expression
for De Donder form is independent of the choice of coordinates. This implies
that the De Donder construction for second order gravity yields a unique
form $\Theta $, which is given by a natural differential operator applied to
the invariant Lagrangian $L$. Therefore, we can use $\Theta $ to obtain an
invariant multisymplectic formulation of second order gravity for any choice
of invariant Lagrangian. \smallskip

The paper is organized as follows. In Section 2, we present a brief review
of some fundamentals of jet bundles. We exhibit the results obtained in this
work in Section 3, where we also discuss multisymplectic formulation of the
second order gravity. Since our proofs are mainly computational and require
a lot of attention to details, they are presented in Section 4.

\section{Geometric background}

\subsection{Jets}

Let $M$ be a $4$-dimensional manifold representing the space-time of general
relativity, and $N\subset \otimes _{\mathrm{sym}}^{2}T^{\ast }M$ be the
bundle of Lorentzian frames on $M$. We denote the canonical projection by $%
\pi :N\rightarrow M.$ A Lorentzian metric on $M$ is a section $\sigma
:M\rightarrow N$ of $\pi :N\rightarrow M$. If $(x^{\mu })=(x^{1},...,x^{4})$
are local coordinates on $M$ with domain $U$, we denote by $(x^{\mu },y_{\mu
\nu })$ the induced coordinates on $\pi ^{-1}(U)\subset N$. In these
coordinates, the section $\sigma $ restricted to $U$ is given by 
\begin{equation}
\sigma _{\mid U}:U\rightarrow \pi ^{-1}(U):x\mapsto (y_{\mu \nu })=g_{\mu
\nu }(x^{\lambda }),  \label{2.1}
\end{equation}%
where $g_{\mu \nu }(x^{\lambda })=g_{\mu \nu }(x^{1},...,x^{4})$ are smooth
functions of the coordinates $(x^{\lambda })$.\footnote{%
Here, we work in the smooth category. In applicatioon to concrete cases, one
has to choose a suitable function space.} \smallskip

The first\ derivatives of sections form the first jet bundle $J^{1}(M,N)$
with the source projection $\pi ^{1}:J^{1}(M,N)\rightarrow M,$ the target
projection $\pi _{0}^{1}:J^{1}(M,N)\rightarrow N$ and the induced
coordinates $(x^{\mu },y_{\mu \nu },z_{\mu \nu \lambda })$ such that 
\begin{eqnarray}
\pi ^{1} &:&J^{1}(M,N)\rightarrow M:(x^{\mu },y_{\mu \nu },z_{\mu \nu
\lambda })\mapsto (x^{\mu })\text{,}  \label{2.2} \\
\pi _{0}^{1} &:&J^{1}(M,N)\rightarrow N:(x^{\mu },y_{\mu \nu },z_{\mu \nu
\lambda })\mapsto (x^{\mu },y_{\mu \nu }).  \notag
\end{eqnarray}%
The first jet extension of the section $\sigma _{\mid U}:U\rightarrow \pi
^{-1}(U)\subseteq N$ is 
\begin{equation}
j^{1}\sigma _{\mid U}:U\rightarrow (\pi ^{1})^{-1}(U)\subseteq
J^{1}(M,N):x\mapsto (x^{\mu },y_{\mu \nu },z_{\mu \nu \lambda })=(x^{\mu
},g_{\mu \nu }(x),g_{\mu \nu ,\lambda }(x)),  \label{2.3}
\end{equation}%
where 
\begin{equation}
g_{\mu \nu ,\lambda }=\frac{\partial }{\partial x^{\lambda }}g_{\mu \nu
}=\partial _{\lambda }g_{\mu \nu }.  \label{2.4}
\end{equation}%
Similarly, for $k=2,3,...$, \ we have the $k$-jet bundle $J^{k}(M,N)$ with
local coordinates $(x^{\mu },y_{\mu \nu },z_{\mu \nu \lambda
_{1}},...,z_{\mu \nu \lambda _{1}...\lambda _{k}})$, source map $\pi
^{k}:J^{k}(M,N)\rightarrow M$, target map $\pi
_{0}^{k}:J^{k}(M,N)\rightarrow N$ and forgetful maps $\pi
_{l}^{k}:J^{k}(M,N)\rightarrow J^{l}(M,N)$ defined for $k>l\geq 0$. The $k$%
-jet extension of a section $\sigma _{\mid U}:U\rightarrow \pi
^{-1}(U)\subseteq N$ is 
\begin{eqnarray}
j^{k}\sigma _{\mid U} &:&U\rightarrow (\pi ^{k})^{-1}(U)\subseteq J^{k}(M,N)
\label{2.5} \\
&:&x\mapsto (x^{\mu },y_{\mu \nu },z_{\mu \nu \lambda _{1}},...,z_{\mu \nu
\lambda _{1}...\lambda _{k}})=(x^{\mu },g_{\mu \nu }(x),g_{\mu \nu ,\lambda
_{1}}(x),...,g_{\mu \nu ,\lambda _{1}...\lambda _{k}}(x)).\smallskip  \notag
\end{eqnarray}

Local contact forms are 
\begin{equation*}
dy_{\mu \nu }-z_{\mu \nu \lambda }dx^{\lambda },~~~dz_{\mu \nu \lambda
}-z_{\mu \nu \lambda \rho }dx^{\rho },~~~...,~~~dz_{\mu \nu \lambda
_{1}...\lambda _{k-1}}-z_{\mu \nu \lambda _{1}...\lambda _{k}}dx^{\lambda
_{k}},
\end{equation*}%
where summation over repeated indices is assumed. A section $\rho
:M\rightarrow J^{k}(M,N)$ of the source projection $\pi ^{k}$ is said to be
holonomic if it is the $k$-jet extension of $\sigma =\pi _{0}^{k}\circ \rho
:M\rightarrow N$. The importance of local contact forms stems from the fact
that a section $\rho :M\rightarrow J^{k}(M,N)$ is holonomic if and only if
the pull-back of every local contact form by $\rho $ vanishes.

Let $Y$ be a vector field on $N$. For every $k\geq 0$, the local 1-parameter
group $\mathrm{e}^{Yt}$ of local diffeomorphisms of $N$ preserving the
projection map $\pi :N\rightarrow M$ gives rise to a local 1-parameter group 
$\mathrm{e}^{Y^{k}t}$ of local diffeomorphisms of $J^{k}(M,N)$, which
preserve the ideal generated by contact forms, and intertwine forgetful
maps. In other words, the following diagram%
\begin{equation*}
\begin{CD} J^k(M,N)@>e^{Y^k t}>>J^k(M,N) \\ @VV{\pi^k_l}V @VV{\pi^k_l}V \\
J^l(M,N) @>e^{Y^l t}>>J^l(M,N) \end{CD}
\end{equation*}%
commutes for $k>l$. The vector field $X^{k}$ on $J^{k}(M,N)$ is called the
prolongation of $X$ to $J^{k}(M,N)$. For more details on jet bundles see 
\cite{olver}.\smallskip

\subsection{Variational problem}

Let $\Lambda $ be the Lagrange form of the second order gravity. This means
that $\Lambda $ is a semi-basic $4$-form on $J^{2}(M,N)$. In local
coordinates, 
\begin{equation}
\Lambda =L(x^{\mu },y_{\mu \nu },z_{\mu \nu \lambda },z_{\mu \nu \lambda
\rho })dx^{1}\wedge ...\wedge dx^{4},  \label{2.6}
\end{equation}%
where $L(x^{\mu },y_{\mu \nu },z_{\mu \nu \lambda },z_{\mu \nu \lambda \rho
})$ is a scalar density with respect to the transformations of $J^{2}(M,N)$
induced by coordinate transformations in $M$. For the sake of simplicity, we
set 
\begin{equation}
L(x^{\mu },y_{\mu \nu },z_{\mu \nu \lambda },z_{\mu \nu \lambda \rho
})=L(x,y,z)\text{ \ and \ }dx^{1}\wedge ...\wedge dx^{4}=d_{4}x,  \label{2.7}
\end{equation}%
so that equation (\ref{2.6}) reads $\Lambda =L(x,y,z)d_{4}x$.\smallskip

Let $\sigma :M\rightarrow N$ be a section of $\pi :N\rightarrow M$. If $%
U\subseteq M$ has compact closure $\overline{U}$, the action $A_{U}$ on $%
\Lambda $ on $\sigma $ is the integral 
\begin{equation}
A_{U}[\sigma ]=\int_{U}(j^{2}\sigma )^{\ast }\Lambda .  \label{2.8}
\end{equation}%
The section $\sigma $ is a critical point of the action $A_{U}$ if, for
every vector field $Y$ on $N$ tangent to the fibres of the source map $\pi
^{2}:J^{2}(M,N)\rightarrow N$ that vanishes on the boundary of $\pi ^{-1}(U)$
up to second order, 
\begin{equation}
\int_{U}(j^{2}\sigma )^{\ast }\pounds _{Y^{2}}\Lambda =0,  \label{2.9}
\end{equation}%
where $Y^{2}$ is the prolongation of $Y$ to $J^{2}(M,N)$ and $\pounds %
_{Y^{2}}\Lambda $ is the Lie derivative of $\Lambda $ with respect to $Y^{2}$%
. The condition that $Y^{2}$ is the prolongation of $Y$ is equivalent to the
classical condition that variations and derivatives commute.

\section{De Donder form}

Following references \cite{js1970b} and \cite{ss}, we present here the
geometric description of the De Donder construction adapted to the second
order gravity.

\begin{definition}
\label{Definition 3.1}De Donder form corresponding to a Lagrangian $\Lambda
=Ld_{4}x$ on $J^{2}(M,N)$ is a form $\Theta $ on $J^{3}(M,N)$ such that, in
local coordinates $(x^{\mu })$ on $M$, 
\begin{eqnarray}
\Theta &=&\pi _{2}^{3\ast }Ld_{4}x+p^{\mu \nu \alpha \beta }(dz_{\mu \nu
\alpha }-z_{\mu \nu \alpha \gamma }dx^{\gamma })\wedge (\partial _{\beta }%
{\mbox{$ \rule {5pt} {.5pt}\rule {.5pt} {6pt} \, $}}%
d_{4}x)+  \label{3.1} \\
&&+p^{\mu \nu \alpha }(dy_{\mu \nu }-z_{\mu \nu \beta }dx^{\beta })\wedge
(\partial _{\alpha }%
{\mbox{$ \rule {5pt} {.5pt}\rule {.5pt} {6pt} \, $}}%
d_{4}x),  \notag
\end{eqnarray}%
where $p^{\mu \nu \alpha \beta }$ and $p^{\mu \nu \alpha }$ are functions on 
$J^{3}(M,N)$ such that, for every local section $\sigma :M\rightarrow
N:x\mapsto (x,g_{\mu \nu }(x))$ of $\pi ,$ 
\begin{eqnarray}
(j^{3}\sigma )^{\ast }p^{\mu \nu \alpha \beta } &=&(j^{3}\sigma )^{\ast }%
\frac{{\small \partial L}}{{\small \partial z}_{\mu \nu \alpha \beta }},
\label{3.2} \\
(j^{3}\sigma )^{\ast }p^{\mu \nu \alpha } &=&(j^{3}\sigma )^{\ast }\frac{%
{\small \partial L}}{{\small \partial z}_{\mu \nu \alpha }}-\frac{{\small %
\partial }}{{\small \partial x}^{\beta }}\left( (j^{3}\sigma )^{\ast }\frac{%
{\small \partial L}}{{\small \partial z}_{\mu \nu \alpha \beta }}\right) . 
\notag
\end{eqnarray}%
\smallskip
\end{definition}

Equations (\ref{3.1}) and (\ref{3.2}) define a $4$-form $\Theta $ on the
domain of coordinates on $J^{3}(M,N)$ defined by local coordinates on $M.$
For a generic Lagrangian, these local forms do not define a global form on $%
J^{3}(M,N).$The main result of our paper is the following theorem\smallskip

\begin{theorem}
\label{Theorem 3.1}For second order gravity with an invariant Lagrangian
form $Ld_{4}x$ on $J^{2}(M,N)$, equations (\ref{3.1}) and (\ref{3.2}) define
a global De Donder form $\Theta $ on $J^{3}(M,N)$, given by a natural
differential operator applied to $Ld_{4}x$.
\end{theorem}

\noindent \textbf{Proof} of Theorem \ref{Theorem 3.1} is given in Section
4.\smallskip

\subsection{Field equations}

Since $\Theta $ differs from $\pi _{2}^{3\ast }Ld_{4}x$ by terms
proportional to contact forms, for every section $\sigma $ of $\pi
:N\rightarrow M$, 
\begin{equation}
A_{U}[\sigma ]=\int_{U}(j^{2}\sigma )^{\ast }\Lambda =\int_{U}(j^{2}\sigma
)^{\ast }Ld_{4}x=\int_{U}(j^{3}\sigma )^{\ast }(\pi _{2}^{3\ast
}Ld_{4}x)=\int_{U}(j^{3}\sigma )^{\ast }\Theta .  \label{3.3}
\end{equation}%
Thus, replacing in the variational principle the Lagrangian form $\Lambda
=Ld_{4}x$ by the corresponding De Donder form does not change the value of
the action $A_{U}[\sigma ]$.\smallskip

\begin{proposition}
\label{Proposition 3.1} For every vector field $X$ on $J^{3}(M,N)$ tangent
to fibres of the target map $\pi _{0}^{3}:J^{3}(M,N)\rightarrow N$ and every
section $\sigma $ of $\pi :N\rightarrow M$, 
\begin{equation}
(j^{3}\sigma )^{\ast }(X%
{\mbox{$ \rule {5pt} {.5pt}\rule {.5pt} {6pt} \, $}}%
d\Theta )=0.  \label{3.4}
\end{equation}
\end{proposition}

\begin{proof}
Equation (\ref{3.1}) yields 
\begin{eqnarray}
&&d\Theta =\pi _{2}^{3\ast }dL\wedge d_{4}x-p^{\mu \nu \alpha \beta }dz_{\mu
\nu \alpha \gamma }\wedge dx^{\gamma }\wedge (\partial _{\beta }%
{\mbox{$ \rule {5pt} {.5pt}\rule {.5pt} {6pt} \, $}}%
d_{4}x)  \label{3.5} \\
&&-p^{\mu \nu \alpha }dz_{\mu \nu \beta }\wedge dx^{\beta }\wedge (\partial
_{\alpha }%
{\mbox{$ \rule {5pt} {.5pt}\rule {.5pt} {6pt} \, $}}%
d_{4}x)+dp^{\mu \nu \alpha \beta }\wedge (dz_{\mu \nu \alpha }-z_{\mu \nu
\alpha \gamma }dx^{\gamma })\wedge (\partial _{\beta }%
{\mbox{$ \rule {5pt} {.5pt}\rule {.5pt} {6pt} \, $}}%
d_{4}x)  \notag \\
&&+dp^{\mu \nu \alpha }\wedge (dy_{\mu \nu }-z_{\mu \nu \beta }dx^{\beta
})\wedge (\partial _{\alpha }%
{\mbox{$ \rule {5pt} {.5pt}\rule {.5pt} {6pt} \, $}}%
d_{4}x).  \notag
\end{eqnarray}%
Consider a vector field in form 
\begin{equation*}
X=X_{\mu \nu \alpha \beta \gamma }\frac{{\small \partial }}{{\small \partial
z}_{\mu \nu \alpha \beta \gamma }}+X_{\mu \nu \alpha \beta }\frac{{\small %
\partial }}{{\small \partial z}_{\mu \nu \alpha \beta }}+X_{\mu \nu \alpha }%
\frac{{\small \partial }}{{\small \partial z}_{\mu \nu \alpha }},
\end{equation*}%
and compute the the following 
\begin{eqnarray*}
(j^{3}\sigma )^{\ast }(X%
{\mbox{$ \rule {5pt} {.5pt}\rule {.5pt} {6pt} \, $}}%
d\Theta ) &=&\left( (j^{3}\sigma )^{\ast }X_{\mu \nu \alpha \beta }\right)
(j^{2}\sigma )^{\ast }\frac{{\small \partial L}}{{\small \partial z}_{\mu
\nu \alpha \beta }}d_{4}x+\left( (j^{3}\sigma )^{\ast }X_{\mu \nu \alpha
}\right) (j^{2}\sigma )^{\ast }\frac{{\small \partial L}}{{\small \partial z}%
_{\mu \nu \alpha }}d_{4}x \\
&&-\left( (j^{3}\sigma )^{\ast }p^{\mu \nu \alpha \beta }\right)
(j^{3}\sigma )^{\ast }X_{\mu \nu \alpha \beta }d_{4}x-\left( (j^{3}\sigma
)^{\ast }p^{\mu \nu \alpha }\right) (j^{3}\sigma )^{\ast }X_{\mu \nu \alpha
}d_{4}x \\
&&+\left( (j^{3}\sigma )^{\ast }dp^{\mu \nu \alpha \beta }\right)
(j^{3}\sigma )^{\ast }X_{\mu \nu \alpha }(\partial _{\beta }%
{\mbox{$ \rule {5pt} {.5pt}\rule {.5pt} {6pt} \, $}}%
d_{4}x).
\end{eqnarray*}%
Here, we used the fact that 
\begin{equation*}
(j^{3}\sigma )^{\ast }\left[ \pi _{2}^{3\ast }dL\right] =(j^{2}\sigma
)^{\ast }dL.
\end{equation*}%
Observe that one has 
\begin{equation*}
\left( (j^{3}\sigma )^{\ast }dp^{\mu \nu \alpha \beta }\right) (j^{3}\sigma
)^{\ast }X_{\mu \nu \alpha }(\partial _{\beta }%
{\mbox{$ \rule {5pt} {.5pt}\rule {.5pt} {6pt} \, $}}%
d_{4}x)=\partial _{\beta }\left( (j^{3}\sigma )^{\ast }p^{\mu \nu \alpha
\beta }\right) (j^{3}\sigma )^{\ast }X_{\mu \nu \alpha }d_{4}x
\end{equation*}%
so that 
\begin{eqnarray*}
(j^{3}\sigma )^{\ast }(X%
{\mbox{$ \rule {5pt} {.5pt}\rule {.5pt} {6pt} \, $}}%
d\Theta ) &=&\left( (j^{3}\sigma )^{\ast }X_{\mu \nu \alpha \beta }\right)
(j^{2}\sigma )^{\ast }\frac{{\small \partial L}}{{\small \partial z}_{\mu
\nu \alpha \beta }}d_{4}x+\left( (j^{3}\sigma )^{\ast }X_{\mu \nu \alpha
}\right) (j^{2}\sigma )^{\ast }\frac{{\small \partial L}}{{\small \partial z}%
_{\mu \nu \alpha }}d_{4}x \\
&&-\left( (j^{3}\sigma )^{\ast }p^{\mu \nu \alpha \beta }\right)
(j^{3}\sigma )^{\ast }X_{\mu \nu \alpha \beta }d_{4}x-\left( (j^{3}\sigma
)^{\ast }p^{\mu \nu \alpha }\right) (j^{3}\sigma )^{\ast }X_{\mu \nu \alpha
}d_{4}x \\
&&+[\partial _{\beta }\left( (j^{3}\sigma )^{\ast }p^{\mu \nu \alpha \beta
}\right) ](j^{3}\sigma )^{\ast }X_{\mu \nu \alpha }d_{4}x \\
&=&\left( (j^{3}\sigma )^{\ast }X_{\mu \nu \alpha \beta }\right) \left[
(j^{2}\sigma )^{\ast }\frac{{\small \partial L}}{{\small \partial z}_{\mu
\nu \alpha \beta }}-\left( (j^{3}\sigma )^{\ast }p^{\mu \nu \alpha \beta
}\right) \right] d_{4}x \\
&&+\left( (j^{3}\sigma )^{\ast }X_{\mu \nu \alpha }\right) \left[
(j^{2}\sigma )^{\ast }\frac{{\small \partial L}}{{\small \partial z}_{\mu
\nu \alpha }}-(j^{3}\sigma )^{\ast }p^{\mu \nu \alpha }+\partial _{\beta
}\left( (j^{3}\sigma )^{\ast }p^{\mu \nu \alpha \beta }\right) \right] d_{4}x
\\
&=&0
\end{eqnarray*}%
by equation (\ref{3.2}).\smallskip\ 
\end{proof}

\begin{lemma}
For each vector field $Y$ on $N$, which projects to a vector field on $M$,
and every section $\sigma $ of $\pi :M\rightarrow N$, 
\begin{equation}
\left( j^{3}\sigma \right) ^{\ast }\left( \pounds _{Y^{3}}\left[ \Theta -\pi
_{2}^{3\ast }Ld_{4}x\right] \right) =0,  \label{3.6}
\end{equation}%
where $Y^{3}$ is the prolongation of $Y$ to $J^{3}(M,N).$\smallskip\ 
\end{lemma}

\begin{proof}
See Lemma 3 in reference \cite{ss}.\smallskip\ 
\end{proof}

Taking these results into account and using Stokes' Theorem, we can rewrite
equation (\ref{2.9}) in the form 
\begin{eqnarray}
\int_{U}(j^{2}\sigma )^{\ast }\pounds _{Y^{2}}\Lambda
&=&\int_{U}(j^{2}\sigma )^{\ast }[\pounds _{Y^{2}}\left( Ld_{4}x)\right)
]=\int_{U}(j^{2}\sigma )^{\ast }[\pounds _{Y^{3}}\left( \pi _{2}^{3\ast
}Ld_{4}x)\right) ]  \notag \\
&=&\int_{U}(j^{3}\sigma )^{\ast }[\pounds _{Y^{3}}\Theta
]=\int_{U}(j^{3}\sigma )^{\ast }\left[ Y^{3}%
{\mbox{$ \rule {5pt} {.5pt}\rule {.5pt} {6pt} \, $}}%
d\Theta +d\left( Y^{3}%
{\mbox{$ \rule {5pt} {.5pt}\rule {.5pt} {6pt} \, $}}%
d\Theta \right) \right]  \notag \\
&=&\int_{U}(j^{3}\sigma )^{\ast }\left[ Y^{3}%
{\mbox{$ \rule {5pt} {.5pt}\rule {.5pt} {6pt} \, $}}%
d\Theta \right] +\int_{\partial U}(j^{3}\sigma )^{\ast }\left[ \left( Y^{3}%
{\mbox{$ \rule {5pt} {.5pt}\rule {.5pt} {6pt} \, $}}%
d\Theta \right) \right]  \notag \\
&=&\int_{U}(j^{3}\sigma )^{\ast }Y^{3}%
{\mbox{$ \rule {5pt} {.5pt}\rule {.5pt} {6pt} \, $}}%
d\Theta =0,  \label{3.7}
\end{eqnarray}%
where $\partial U$ is the boundary of $U$, and the integral over the
boundary vanishes because we assume that $Y$ vanishes on $\partial U$ to
second order. Proposition \ref{Proposition 3.1} implies that in equation (%
\ref{3.7}) we may replace the prolongation $Y^{3}$ of a vector field $Y$ on $%
N$ tangent to fibres of $\pi :M\rightarrow N$ by arbitrary vector field $X$
on $J^{3}(M,N)$ tangent to fibres of the source map $\pi
^{3}:J^{3}(M,N)\rightarrow M$ and vanishes on $(\pi ^{3})^{-1}(\partial U).$
Therefore, the variational principle (\ref{2.9}) is equivalent to 
\begin{equation}
\int_{U}(j^{3}\sigma )^{\ast }\left[ X%
{\mbox{$ \rule {5pt} {.5pt}\rule {.5pt} {6pt} \, $}}%
d\Theta \right] =0,  \label{3.8}
\end{equation}%
where $X$ is an arbitrary vector field on $J^{3}(M,N)$ tangent to fibres of
the source map $\pi ^{3}:J^{3}(M,N)\rightarrow M$. The Fundamental Theorem
in the Calculus of Variations ensures that the variational principle (\ref%
{3.8}) is equivalent to 
\begin{equation}
(j^{3}\sigma _{\mid U})^{\ast }\left[ X%
{\mbox{$ \rule {5pt} {.5pt}\rule {.5pt} {6pt} \, $}}%
d\Theta \right] =0  \label{3.9}
\end{equation}%
for every vector field $X$ on $J^{3}(M,N)$ tangent to fibres of the source
map $\pi ^{3}:J^{3}(M,N)\rightarrow M.$ Equation (\ref{3.9}) is the De
Donder equation for the second order gravity with invariant Lagrangian $L$.
\smallskip

We can show directly that equation (\ref{3.9}) is a system of equations in
differential forms equivalent to the Euler-Lagrange equations corresponding
to $L$. Let 
\begin{equation*}
X=X_{\mu \nu \alpha \beta \gamma }\frac{\partial }{\partial z_{\mu \nu
\alpha \beta \gamma }}+X_{\mu \nu \alpha \beta }\frac{\partial }{\partial
z_{\mu \nu \alpha \beta }}+X_{\mu \nu \alpha }\frac{\partial }{\partial
z_{\mu \nu \alpha }}+X_{\mu \nu }\frac{\partial }{\partial y_{\mu \nu }}
\end{equation*}%
be a vector field tangent to fibres of the source map, and let $\sigma
:(x^{\lambda })\mapsto g_{\mu \nu }(x^{\lambda })$ be a section of $\pi
:N\rightarrow M$. Introducing the notation 
\begin{equation}
P^{\mu \nu \alpha \beta }=(j^{3}\sigma )^{\ast }p^{\mu \nu \alpha \beta 
\text{ }}\text{ and \ }P^{\mu \nu \alpha }=(j^{3}\sigma )^{\ast }p^{\mu \nu
\alpha \text{ }}.  \label{3.5x}
\end{equation}%
we can write the left hand side of equation (\ref{3.4}) in the from 
\begin{eqnarray*}
&&(j^{3}\sigma )^{\ast }\left[ X%
{\mbox{$ \rule {5pt} {.5pt}\rule {.5pt} {6pt} \, $}}%
d\Theta \right] \\
&=&\left[ X_{\mu \nu \alpha \beta }\left( (j^{2}\sigma )^{\ast }\frac{%
\partial L}{\partial z_{\mu \nu \alpha \beta }}-P^{\mu \nu \alpha \beta
}\right) +X_{\mu \nu \alpha }\left( (j^{2}\sigma )^{\ast }\frac{\partial L}{%
\partial z_{\mu \nu \alpha }}-P^{\mu \nu \alpha }\right) \right]
d_{4}x \\
&&+\left[X_{\mu \nu}(j^{2}\sigma )^{\ast }\frac{\partial L}{\partial y_{\mu \nu }}\right]
d_{4}x -\left[ X_{\mu \nu \alpha }P_{,\beta }^{\mu \nu \alpha \beta }+X_{\mu \nu
}P_{,\alpha }^{\mu \nu \alpha }\right] d_{4}x.
\end{eqnarray*}%
Since components of $X$ are arbitrary, equation (\ref{3.4}) reads%
\begin{eqnarray}
(j^{2}\sigma )^{\ast }\frac{\partial L}{\partial z_{\mu \nu \alpha \beta }}%
-P^{\mu \nu \alpha \beta } &=&0,  \label{3.5e} \\
(j^{2}\sigma )^{\ast }\frac{\partial L}{\partial z_{\mu \nu \alpha }}-P^{\mu
\nu \alpha }-P_{,\beta }^{\mu \nu \alpha \beta } &=&0,  \label{3.5f} \\
(j^{2}\sigma )^{\ast }\frac{\partial L}{\partial y_{\mu \nu }}-P_{,\alpha
}^{\mu \nu \alpha } &=&0.  \label{3.5g}
\end{eqnarray}%
Equation (\ref{3.5e}) is the definition of $P^{\mu \nu \alpha \beta }$,
equation (\ref{3.5f}) is the definition of $P^{\mu \nu \alpha }$, while
equation (\ref{3.5g}) is equivalent to the Euler-Lagrange equations 
\begin{equation}
(j^{2}\sigma )^{\ast }\frac{\partial L}{\partial y_{\mu \nu }}-\frac{%
\partial }{\partial x^{\alpha }}\left( (j^{2}\sigma )^{\ast }\frac{\partial L%
}{\partial z_{\mu \nu \alpha }}\right) +\frac{\partial ^{2}}{\partial
x^{\alpha }\partial x^{\beta }}\left( (j^{2}\sigma )^{\ast }\frac{\partial L%
}{\partial z_{\mu \nu \alpha \beta }}\right) =0.  \label{3.5h}
\end{equation}%
\smallskip

\subsection{Example: Hilbert's Lagrangian}

Hilbert's Lagrangian of general relativity, expressed in terms of local
coordinates, is \ 
\begin{equation}
L_{\mathrm{Hilbert}}=R[g]\sqrt{-\det g},  \label{3.9a}
\end{equation}%
where $R[g]$ is the scalar curvature of the Lorentzian metric $g$. Since $L_{%
\mathrm{Hilbert}}$ depends linearly on second derivatives of the metric, the
corresponding Euler-Lagrange equations are of second order. The Arnowitt,
Deser and Misner Hamiltonian formalism for general relativity, \cite{adm},
see also \cite{mtw}, is based on the Palatini formalism, \cite{palatini}, in
which metric and connection are independent dynamical variables. The De
Donder form for the Palatini formulation of general relativity was given in 
\cite{js1970}. \smallskip

\begin{proposition}
\label{Proposition 3.3}The De Donder form for the second order Hilbert
Lagrangian $L_{\mathrm{Hilbert}}$, expressed in local coordinates, is%
\begin{eqnarray}
\Theta _{\mathrm{Hilbert}} &=&\Gamma _{\lambda \mu \nu }\Gamma _{\alpha
\beta \gamma }\left( g^{\alpha \gamma }g^{\mu \nu }g^{\beta \lambda
}-g^{\beta \lambda }g^{\alpha \mu }g^{\gamma \nu }\right) \sqrt{-\det g}%
\mathrm{d}_{4}x  \label{3.9b} \\
&&+\frac{1}{2}\left( -g^{\alpha \mu }\Gamma ^{\beta }-g^{\beta \mu }\Gamma
^{\alpha }+\Gamma ^{\beta \alpha \mu }+\Gamma ^{\alpha \beta \mu }\right) 
\sqrt{-\det g}\mathrm{d}g_{\alpha \beta }\wedge \left( \frac{{\small %
\partial }}{{\small \partial x}^{\mu }}%
{\mbox{$ \rule {5pt} {.5pt}\rule {.5pt} {6pt} \, $}}%
\mathrm{d}_{4}x\right)  \notag \\
&&+\frac{1}{2}\left( g^{\alpha \mu }g^{\beta \nu }+g^{\alpha \nu }g^{\beta
\mu }-2g^{\alpha \beta }g^{\mu \nu }\right) \sqrt{-\det g}\mathrm{d}%
g_{\alpha \beta ,\nu }\wedge \left( \frac{{\small \partial }}{{\small %
\partial x}^{\mu }}%
{\mbox{$ \rule {5pt} {.5pt}\rule {.5pt} {6pt} \, $}}%
\mathrm{d}_{4}x\right) ,  \notag
\end{eqnarray}%
where $\Gamma ^{\alpha }=g^{\mu \nu }\Gamma _{\mu \nu }^{\alpha }.$
\end{proposition}

\noindent \textbf{Proof} of Proposition \ref{Proposition 3.3} is given in
Section 4.

\subsection{Example: Matter and Gravitation}

In the study of second order gravity, we cannot ignore the interaction of
gravity with matter. If $L$ is the Lagrangian for gravity alone and $L_{%
\mathrm{matter}}$ is the Lagrangian for the matter such that the total
Lagrangian $L_{\mathrm{total}}=L+L_{\mathrm{matter}}$ is invariant under the
group of diffeomorphisms of $M$, we conjecture that the statement of Theorem %
\ref{Theorem 3.1} also holds for the De Donder form $\Theta _{\mathrm{total}%
} $ associated to the total Lagrangian $L_{\mathrm{total}}$. Here, we
illustrate it with the case of second order gravity interacting with a
scalar field $\phi $ given by the Lagrangian form 
\begin{equation}
L_{\mathrm{matter}}[g,\phi ]d_{4}x=\left[ \frac{{\small 1}}{{\small 2}}%
g^{\mu \nu }\frac{\partial \phi }{\partial x^{\mu }}\frac{\partial \phi }{%
\partial x^{\nu }}+V(\phi )\right] \sqrt{-\det g}\mathrm{d}_{4}x.
\label{3.10}
\end{equation}%
where $V(\phi )$ is a known function. We may consider matter field to be a
section $\phi $ of the trivial bundle 
\begin{equation}
\kappa :\mathbb{R}\times M\rightarrow M:(t,x)\mapsto x.  \label{3.11}
\end{equation}%
We understand, the Lagrangian $L_{\mathrm{matter}}$, exhibited in (\ref{3.10}%
), as of first order in $\phi $ and depends parametrically the Lorentzian
metric $g$. Introducing local coordinates $(x^{\mu },t,z_{\mu })$ on the
first jet bundle $J^{1}(M,\mathbb{R}\times M)$, we write the contact form as 
$dt-z_{\mu }dx^{\mu }$, and $g$-dependent function $L_{\mathrm{matter}}$ as%
\begin{equation}
L_{\mathrm{matter}}^{g}(x,t,z_{\mu })=\left[ \frac{{\small 1}}{{\small 2}}%
g^{\mu \nu }z_{\mu }z_{\nu }+V(z)\right] \sqrt{-\det g}.  \label{3.11a}
\end{equation}%
In the coordinate representation, the De Donder form for the present case is
defined to be 
\begin{equation}
\Theta _{\mathrm{matter}}^{g}=L_{\mathrm{matter}}^{g}\mathrm{d}_{4}x+q^{\mu
}(\mathrm{d}t-z_{\nu }\mathrm{d}x^{\nu })\wedge (\partial _{\mu }%
{\mbox{$ \rule {5pt} {.5pt}\rule {.5pt} {6pt} \, $}}%
\mathrm{d}_{4}x),  \label{3.12}
\end{equation}%
where $q^{\mu }$ is a function on $J^{1}(M,K)$ such that, for every local
section $\phi $ of the trivial fibration $\kappa $,%
\begin{equation}
(j^{1}\phi )^{\ast }q^{\mu }=(j^{1}\phi )^{\ast }\frac{{\small \partial }}{%
{\small \partial z}_{\mu }}L_{\mathrm{matter}}^{g}=g^{\mu \nu }z_{\nu }\sqrt{%
-\det g}.  \label{3.13}
\end{equation}%
The total space is the fibre product $J^{3}(M,N)\times _{M}J^{1}(M,K)$ of
the fibrations $\pi ^{3}:J^{3}(M,N)\rightarrow M$ and $\kappa
^{1}:J^{1}(M,K)\rightarrow M$. The total De Donder form $\Theta _{\mathrm{%
total}}$ is the pull back of the sum of the De Donder form $\Theta $ for the
gravity and the De Donder form $\Theta _{\mathrm{matter}}^{g}$ of the matter
to the this Whitney product.\ In coordinates, the total De Donder form is
given by%
\begin{eqnarray}
\Theta _{\mathrm{total}} &=&L\mathrm{d}_{4}x+p^{\mu \nu \alpha \beta
}(dz_{\mu \nu \alpha }-z_{\mu \nu \alpha \gamma }dx^{\gamma })\wedge
(\partial _{\beta }%
{\mbox{$ \rule {5pt} {.5pt}\rule {.5pt} {6pt} \, $}}%
d_{4}x)+  \label{3.14} \\
&&+p^{\mu \nu \alpha }(dy_{\mu \nu }-z_{\mu \nu \beta }dx^{\beta })\wedge
(\partial _{\alpha }%
{\mbox{$ \rule {5pt} {.5pt}\rule {.5pt} {6pt} \, $}}%
d_{4}x)+  \notag \\
&&+L_{\mathrm{matter}}^{y}d_{4}x+q^{\mu }(dt-z_{\nu }dx^{\nu })\wedge
(\partial _{\mu }%
{\mbox{$ \rule {5pt} {.5pt}\rule {.5pt} {6pt} \, $}}%
d_{4}x),  \notag
\end{eqnarray}%
where, for every local section 
\begin{equation*}
\sigma \times \phi :M\rightarrow N\times _{M}(\mathbb{R}\times M):x\mapsto
(x,g_{\mu \nu }(x),\phi (x))
\end{equation*}%
of the fibration $\pi \times _{M}\kappa $ from the product manifold $N\times
_{M}(\mathbb{R}\times M)$ to the base manifold $M$, the coefficient
functions are 
\begin{eqnarray}
(j^{3}\sigma )^{\ast }p^{\mu \nu \alpha \beta } &=&(j^{2}\sigma )^{\ast }%
\frac{{\small \partial L}}{{\small \partial z}_{\mu \nu \alpha \beta }},
\label{3.15} \\
(j^{3}\sigma )^{\ast }p^{\mu \nu \alpha } &=&(j^{2}\sigma )^{\ast }\frac{%
{\small \partial L}}{{\small \partial z}_{\mu \nu \alpha }}-\frac{{\small %
\partial }}{{\small \partial x}^{\beta }}\left( (j^{2}\sigma )^{\ast }\frac{%
{\small \partial L}}{{\small \partial z}_{\mu \nu \alpha \beta }}\right) ,
\label{3.16} \\
(j^{1}\phi )^{\ast }q^{\mu } &=&g^{\mu \nu }\phi _{,\nu }\sqrt{-\det g}.
\label{3.17}
\end{eqnarray}%
In equation (\ref{3.14}) we did not put any pull-back signs to make the
coordinate picture more transparent. Moreover, we replaced the subscript $g$
over $L_{\mathrm{matter}}^{g}$ in equation (3.13) by $y$ in order to
indicate dependence of $L_{\mathrm{matter}}$ on the variable $y_{\mu \nu
}=g_{\mu \nu }(x)$.

\begin{proposition}
\label{Proposition 3.2} For second order gravity with invariant Lagrangian $%
L $ on $J^{2}(M,N)$, the expressions in \ equations (\ref{3.14}) and (\ref%
{3.15}-\ref{3.17}) are independent of the choice of coordinates $(x^{\mu })$
on $M$. Hence, they define a global $4$-form $\Theta _{\mathrm{total}}$ on
the fibre product $J^{3}(M,N)\times _{M}J^{1}(M,K).$
\end{proposition}

\noindent \textbf{Proof }of Proposition \ref{Proposition 3.2} is given in
Section 4.\smallskip

As in preceeding section, the Euler-Lagrange equations for the total
Lagrangian $L_{\mathrm{total}}$ are equivalent to the De Donder equations
for the total De Donder form $\Theta _{\mathrm{total}}$. \smallskip

\section{Proofs}

\subsection{Proof of Theorem 2}

Recall that De Donder form corresponding to a Lagrangian form $\Lambda
=Ld_{4}x$ on $J^{2}(M,N)$ is a form $\Theta $ on $J^{3}(M,N),$ defined by 
\begin{eqnarray}
\Theta &=&\pi _{2}^{3\ast }L\mathrm{d}_{4}x+p^{\mu \nu \alpha \beta }(%
\mathrm{d}z_{\mu \nu \alpha }-z_{\mu \nu \alpha \gamma }\mathrm{d}x^{\gamma
})\wedge (\partial _{\beta }{%
\mbox{$ \rule {5pt} {.5pt}\rule
{.5pt} {6pt} \, $}}d_{4}x)  \label{4.1} \\
&&+p^{\mu \nu \alpha }(\mathrm{d}y_{\mu \nu }-z_{\mu \nu \beta }\mathrm{d}%
x^{\beta })\wedge (\partial _{\alpha }{%
\mbox{$ \rule {5pt} {.5pt}\rule
{.5pt} {6pt} \, $}}d_{4}x),  \notag
\end{eqnarray}%
in the domain of the coordinate chart on $J^{3}(M,N)$ induced by a
coordinate chart $(x^{\mu })$ on $M,$ where $\pi _{2}^{3\ast }L\mathrm{d}%
_{4}x$ is the pull-back of the Lagrangian form by the forgetful map $\pi
_{2}^{3}:J^{3}(M,N)\rightarrow J^{2}(M,N)$, while $p^{\mu \nu \alpha \beta }$
and $p^{\mu \nu \alpha }$ are Ostrogradski's momenta. In other words, $%
p^{\mu \nu \alpha \beta }$ and $p^{\mu \nu \alpha }$ are functions on $%
J^{3}(M,N)$ such that 
\begin{equation}
\begin{split}
(j^{3}\sigma )^{\ast }p^{\mu \nu \alpha \beta }& =(j^{2}\sigma )^{\ast }%
\frac{{\small \partial L}}{{\small \partial z}_{\mu \nu \alpha \beta }}, \\
(j^{3}\sigma )^{\ast }p^{\mu \nu \alpha }& =(j^{2}\sigma )^{\ast }\frac{%
{\small \partial L}}{{\small \partial z}_{\mu \nu \alpha }}-\frac{{\small %
\partial }}{{\small \partial x}^{\beta }}\left( (j^{2}\sigma )^{\ast }\frac{%
{\small \partial L}}{{\small \partial z}_{\mu \nu \alpha \beta }}\right) ,
\end{split}
\label{4.2}
\end{equation}%
for every local section $\sigma $ of $\pi $.

Our aim in this section is to study transformation laws of components of $%
\Theta $ with respect to coordinate transformations in $J^{3}(M,N)$ induced
by an orientation preserving coordinate transformation 
\begin{equation}
(x^{\mu ^{\prime }})\rightarrow (x^{\mu })=(x^{\mu }(x^{\mu ^{\prime }})).
\label{gM}
\end{equation}%
on $M$. It induces a local coordinate transformation on $N$ given by%
\begin{equation}
(x^{\mu ^{\prime }},y_{\mu ^{\prime }\nu ^{\prime }})\rightarrow (x^{\mu
},y_{\mu \nu }=y_{\mu ^{\prime }\nu ^{\prime }}x_{,\mu }^{\mu ^{\prime
}}x_{,\nu }^{\nu ^{\prime }}).  \label{gN}
\end{equation}%
Further, we have the local expressions for transformations for the jet
coordinates 
\begin{align}
z_{\mu ^{\prime }\nu ^{\prime }\alpha ^{\prime }} &=z_{\mu \nu \alpha
}x_{,\mu ^{\prime }}^{\mu }x_{,\nu ^{\prime }}^{\nu }x_{,\alpha ^{\prime
}}^{\alpha }+y_{\mu \nu }x_{,\mu ^{\prime }\alpha ^{\prime }}^{\mu }x_{,\nu
^{\prime }}^{\nu }+y_{\mu \nu }x_{,\mu ^{\prime }}^{\mu }x_{,\nu ^{\prime
}\alpha ^{\prime }}^{\nu }  \label{prolonged-g-1} \\
z_{\alpha ^{\prime }\beta ^{\prime }\gamma ^{\prime }\lambda ^{\prime }}
&=z_{\alpha \beta \gamma \lambda }x_{,\alpha ^{\prime }}^{\alpha }x_{,\beta
^{\prime }}^{\beta }x_{,\gamma ^{\prime }}^{\gamma }x_{,\lambda ^{\prime
}}^{\lambda }+z_{\alpha \beta \gamma }x_{,\alpha ^{\prime }\lambda ^{\prime
}}^{\alpha }x_{,\beta ^{\prime }}^{\beta }x_{,\gamma ^{\prime }}^{\gamma
}+z_{\alpha \beta \gamma }x_{,\alpha ^{\prime }}^{\alpha }x_{,\beta ^{\prime
}\lambda ^{\prime }}^{\beta }x_{,\gamma ^{\prime }}^{\gamma }+z_{\alpha
\beta \gamma }x_{,\alpha ^{\prime }}^{\alpha }x_{,\beta ^{\prime }}^{\beta
}x_{,\gamma ^{\prime }\lambda ^{\prime }}^{\gamma }  \notag \\
&+z_{\alpha \beta \lambda }x_{,\lambda ^{\prime }}^{\lambda }x_{,\alpha
^{\prime }\gamma ^{\prime }}^{\alpha }x_{,\beta ^{\prime }}^{\beta
}+z_{\alpha \beta \lambda }x_{,\lambda ^{\prime }}^{\lambda }x_{,\alpha
^{\prime }}^{\alpha }x_{,\beta ^{\prime }\gamma ^{\prime }}^{\beta
}+y_{\alpha \beta }x_{,\alpha ^{\prime }\gamma ^{\prime }\lambda ^{\prime
}}^{\alpha }x_{,\beta ^{\prime }}^{\beta }+y_{\alpha \beta }x_{,\alpha
^{\prime }\gamma ^{\prime }}^{\alpha }x_{,\beta ^{\prime }\lambda ^{\prime
}}^{\beta }  \notag \\
&+y_{\alpha \beta }x_{,\alpha ^{\prime }\lambda ^{\prime }}^{\alpha
}x_{,\beta ^{\prime }\gamma ^{\prime }}^{\beta }+y_{\alpha \beta }x_{,\alpha
^{\prime }}^{\alpha }x_{,\beta ^{\prime }\gamma ^{\prime }\lambda ^{\prime
}}^{\beta }.  \label{prolonged-g}
\end{align}

By assumption the Lagrangian form $\Lambda =Ld_{4}x$ is invariant under the
transformations (\ref{gM}) through (\ref{prolonged-g}). This implies that
under these transformations, $L$ transforms as a scalar density. In other
words, 
\begin{equation}
L\left( x^{\mu ^{\prime }},y_{\mu ^{\prime }\nu ^{\prime }},z_{\mu ^{\prime
}\nu ^{\prime }\alpha ^{\prime },}z_{\mu ^{\prime }\nu ^{\prime }\alpha
^{\prime }\beta ^{\prime }}\right) \det (x_{,\lambda }^{\lambda ^{\prime
}})=L\left( x^{\mu },y_{\mu \nu },z_{\mu \nu \alpha ,}z_{\mu \nu \alpha
\beta }\right) .  \label{inv-Lag}
\end{equation}

\begin{lemma}
\label{Lemma-4.2} Since the boundary form 
\begin{equation}
\Xi =p^{\mu \nu \alpha \beta }(\mathrm{d}z_{\mu \nu \alpha }-z_{\mu \nu
\alpha \gamma }\mathrm{d}x^{\gamma })\wedge (\partial _{\beta }%
{\mbox{$ \rule {5pt} {.5pt}\rule {.5pt} {6pt} \, $}}%
\mathrm{d}_{4}x)+p^{\mu \nu \alpha }(\mathrm{d}y_{\mu \nu }-z_{\mu \nu \beta
}\mathrm{d}x^{\beta })\wedge (\partial _{\alpha }%
{\mbox{$ \rule {5pt} {.5pt}\rule {.5pt} {6pt} \, $}}%
\mathrm{d}_{4}x)  \label{Zeta}
\end{equation}%
is defined to be 4-form on $J^{3}(M,N),$ under the change of coordinates (%
\ref{gM}) through (\ref{prolonged-g}), the coefficients $p^{\mu \nu \alpha }$
and $p^{\mu \nu \alpha \beta }$ transform as follows\footnote{%
Since $\Xi $ depends on the third jet variables only through $p^{\mu \nu
\alpha },$ we need not write explicit transformation rules for the third
jets.}%
\begin{eqnarray}
p^{\mu \nu \alpha } &=&\left( p^{\mu ^{\prime }\nu ^{\prime }\alpha ^{\prime
}}x_{,\mu ^{\prime }}^{\mu }y_{,\nu ^{\prime }}^{\nu }x_{,\alpha ^{\prime
}}^{\alpha }+p^{\mu ^{\prime }\nu ^{\prime }\alpha ^{\prime }\beta ^{\prime
}}x_{,\beta ^{\prime }}^{\beta }x_{,\alpha ^{\prime }}^{\alpha }\left(
x_{,\mu ^{\prime }\beta }^{\mu }x_{,\nu ^{\prime }}^{\nu }+x_{,\mu ^{\prime
}}^{\mu }x_{,\nu ^{\prime }\beta }^{\nu }\right) \right) \det (x_{,\lambda
}^{\lambda ^{\prime }}),  \label{p-1} \\
p^{\mu \nu \alpha \beta } &=&p^{\mu ^{\prime }\nu ^{\prime }\alpha ^{\prime
}\beta ^{\prime }}x_{,\mu ^{\prime }}^{\mu }x_{,\nu ^{\prime }}^{\nu
}x_{,\alpha ^{\prime }}^{\alpha }x_{,\beta ^{\prime }}^{\beta }\det
(x_{,\lambda }^{\lambda ^{\prime }}).  \label{p-2}
\end{eqnarray}
\end{lemma}

\begin{proof}
Notice that, the boundary term $\Xi $ in (\ref{Zeta}) is the sum of two four
forms, label them as $\Xi _{1}$ and $\Xi _{2}$ in a respected order. In
order to deduce the transformation properties of the coefficients $p^{\mu
\nu \alpha }$ and $p^{\mu \nu \alpha \beta }$, we express $\Xi _{1}$ and $%
\Xi _{2}$ in primed coordinates using the transformations (\ref{gM}) through
(\ref{prolonged-g}) under the assumption that $\Xi _{1}$ and $\Xi _{2}$ are
4-forms on $J^{3}(M,N)$ .

Consider first the term $\Xi _{1}=p^{\mu \nu \alpha \beta }(\mathrm{d}z_{\mu
\nu \alpha }-z_{\mu \nu \alpha \gamma }\mathrm{d}x^{\gamma })\wedge
(\partial _{\beta }%
{\mbox{$ \rule {5pt} {.5pt}\rule {.5pt} {6pt} \, $}}%
\mathrm{d}_{4}x)$. Taking exterior differential of the coordinate
transformation (\ref{prolonged-g-1}), we get 
\begin{eqnarray}
\mathrm{d}z_{\mu ^{\prime }\nu ^{\prime }\alpha ^{\prime }} &=&x_{,\mu
^{\prime }}^{\mu }x_{,\nu ^{\prime }}^{\nu }x_{,\alpha ^{\prime }}^{\alpha }%
\mathrm{d}z_{\mu \nu \alpha }+z_{\mu \nu \alpha }x_{,\nu ^{\prime }}^{\nu
}x_{,\alpha ^{\prime }}^{\alpha }\mathrm{d}x_{,\mu ^{\prime }}^{\mu }+z_{\mu
\nu \alpha }x_{,\mu ^{\prime }}^{\mu }x_{,\alpha ^{\prime }}^{\alpha }%
\mathrm{d}x_{,\nu ^{\prime }}^{\nu }  \notag \\
&&+z_{\mu \nu \alpha }x_{,\mu ^{\prime }}^{\mu }x_{,\nu ^{\prime }}^{\nu }%
\mathrm{d}x_{,\alpha ^{\prime }}^{\alpha }+x_{,\mu ^{\prime }\alpha ^{\prime
}}^{\mu }x_{,\nu ^{\prime }}^{\nu }\mathrm{d}y_{\mu \nu }+y_{\mu \nu
}x_{,\nu ^{\prime }}^{\nu }\mathrm{d}x_{,\mu ^{\prime }\alpha ^{\prime
}}^{\mu }+y_{\mu \nu }x_{,\mu ^{\prime }\alpha ^{\prime }}^{\mu }\mathrm{d}%
x_{,\nu ^{\prime }}^{\nu }  \notag \\
&&+x_{,\mu ^{\prime }}^{\mu }x_{,\nu ^{\prime }\alpha ^{\prime }}^{\nu }%
\mathrm{d}y_{\mu \nu }+y_{\mu \nu }x_{,\mu ^{\prime }}^{\mu }\mathrm{d}%
x_{,\nu ^{\prime }\alpha ^{\prime }}^{\nu }+y_{\mu \nu }x_{,\nu ^{\prime
}\alpha ^{\prime }}^{\nu }\mathrm{d}x_{,\mu ^{\prime }}^{\mu }.  \label{dz}
\end{eqnarray}%
Recalling some basic chain rule operations 
\begin{equation*}
x_{,\gamma ^{\prime }}^{\gamma }\mathrm{d}x^{\gamma ^{\prime }}=\mathrm{d}%
x^{\gamma },\text{ \ \ }\mathrm{d}x_{,\mu ^{\prime }\alpha ^{\prime }}^{\mu
}=x_{,\mu ^{\prime }\alpha ^{\prime }\gamma ^{\prime }}^{\mu }\mathrm{d}%
x^{\gamma ^{\prime }},
\end{equation*}%
and using the transformation rule (\ref{prolonged-g}) of $z_{\mu ^{\prime
}\nu ^{\prime }\alpha ^{\prime }\gamma ^{\prime }}$, one computes%
\begin{eqnarray}
&&z_{\mu ^{\prime }\nu ^{\prime }\alpha ^{\prime }\gamma ^{\prime }}\mathrm{d%
}x^{\gamma ^{\prime }}=z_{\mu \nu \alpha \gamma }x_{,\mu ^{\prime }}^{\mu
}x_{,\nu ^{\prime }}^{\nu }x_{,\alpha ^{\prime }}^{\alpha }\mathrm{d}%
x^{\gamma }+z_{\mu \nu \alpha }x_{,\nu ^{\prime }}^{\nu }x_{,\alpha ^{\prime
}}^{\alpha }\mathrm{d}x_{,\mu ^{\prime }}^{\mu }+z_{\mu \nu \alpha }x_{,\mu
^{\prime }}^{\mu }x_{,\alpha ^{\prime }}^{\alpha }\mathrm{d}x_{,\nu ^{\prime
}}^{\nu }  \notag \\
&&+z_{\mu \nu \alpha }x_{,\mu ^{\prime }}^{\mu }x_{,\nu ^{\prime }}^{\nu }%
\mathrm{d}x_{,\alpha ^{\prime }}^{\alpha }+z_{\mu \nu \gamma }x_{,\mu
^{\prime }\alpha ^{\prime }}^{\mu }x_{,\nu ^{\prime }}^{\nu }\mathrm{d}%
x^{\gamma }+y_{\mu \nu }x_{,\nu ^{\prime }}^{\nu }\mathrm{d}x_{,\mu ^{\prime
}\alpha ^{\prime }}^{\mu }+y_{\mu \nu }x_{,\mu ^{\prime }\alpha ^{\prime
}}^{\mu }\mathrm{d}x_{,\nu ^{\prime }}^{\nu }  \notag \\
&&+z_{\mu \nu \gamma }x_{,\mu ^{\prime }}^{\mu }x_{,\nu ^{\prime }\alpha
^{\prime }}^{\nu }\mathrm{d}x^{\gamma }+y_{\mu \nu }x_{,\nu ^{\prime }\alpha
^{\prime }}^{\nu }\mathrm{d}x_{,\mu ^{\prime }}^{\mu }+y_{\mu \nu }x_{,\mu
^{\prime }}^{\mu }\mathrm{d}x_{,\nu ^{\prime }\alpha ^{\prime }}^{\nu }.
\label{zdx}
\end{eqnarray}%
We subtract the one-form $z_{\mu ^{\prime }\nu ^{\prime }\alpha ^{\prime
}\gamma ^{\prime }}\mathrm{d}x^{\gamma ^{\prime }}$, exhibited in (\ref{zdx}%
), from of the one-from $\mathrm{d}z_{\mu ^{\prime }\nu ^{\prime }\alpha
^{\prime }}$ in (\ref{dz}). While taking the difference, see that the
second, the third, and the fourth terms in the first line of (\ref{dz})
cancel with the second, the third, and the fourth terms in the first line of
(\ref{zdx}), respectively. Notice also that the second and the third terms
both in the the second and the third lines of (\ref{dz})\ cancel with the
second and the third terms of the second and third lines of (\ref{dz}),
respectively. Eventually, we arrive at the following expression%
\begin{eqnarray*}
\mathrm{d}z_{\mu ^{\prime }\nu ^{\prime }\alpha ^{\prime }}-z_{\mu ^{\prime
}\nu ^{\prime }\alpha ^{\prime }\gamma ^{\prime }}\mathrm{d}x^{\gamma
^{\prime }} &=&x_{,\mu ^{\prime }}^{\mu }x_{,\nu ^{\prime }}^{\nu
}x_{,\alpha ^{\prime }}^{\alpha }\left( \mathrm{d}z_{\mu \nu \alpha }-z_{\mu
\nu \alpha \gamma }\mathrm{d}x^{\gamma }\right) \\
&&+\left( x_{,\mu ^{\prime }\alpha ^{\prime }}^{\mu }x_{,\nu ^{\prime
}}^{\nu }+x_{,\mu ^{\prime }}^{\mu }x_{,\nu ^{\prime }\alpha ^{\prime
}}^{\nu }\right) \left( \mathrm{d}y_{\mu \nu }-z_{\mu \nu \gamma }\mathrm{d}%
x^{\gamma }\right) .
\end{eqnarray*}%
On the other hand, it is immediate to see that 
\begin{equation}
\partial _{\beta ^{\prime }}%
{\mbox{$ \rule {5pt} {.5pt}\rule {.5pt} {6pt} \, $}}%
d_{4}x^{\prime }=x_{,\beta ^{\prime }}^{\beta }\det (x_{,\lambda }^{\lambda
^{\prime }})\partial _{\beta }%
{\mbox{$ \rule {5pt} {.5pt}\rule {.5pt} {6pt} \, $}}%
d_{4}x.  \label{Wedge-trf}
\end{equation}%
Hence the first term in the boundary form can be obtained by first taking
the exterior product of the one form $dz_{\mu ^{\prime }\nu ^{\prime }\alpha
^{\prime }}-z_{\mu ^{\prime }\nu ^{\prime }\alpha ^{\prime }\gamma ^{\prime
}}dx^{\gamma ^{\prime }}$ and the three form $\partial _{\beta ^{\prime }}%
{\mbox{$ \rule {5pt} {.5pt}\rule {.5pt} {6pt} \, $}}%
d_{4}x^{\prime }$ then by multiplying the product with $p^{\mu ^{\prime }\nu
^{\prime }\alpha ^{\prime }\beta ^{\prime }}$. This shows that $\Xi _{1}$
expressed in terms of the primed coordinates is 
\begin{eqnarray}
\Xi _{1}^{\prime } &=&p^{\mu ^{\prime }\nu ^{\prime }\alpha ^{\prime }\beta
^{\prime }}(dz_{\mu ^{\prime }\nu ^{\prime }\alpha ^{\prime }}-z_{\mu
^{\prime }\nu ^{\prime }\alpha ^{\prime }\gamma ^{\prime }}dx^{\gamma
^{\prime }})\wedge (\partial _{\beta ^{\prime }}%
{\mbox{$ \rule {5pt} {.5pt}\rule {.5pt} {6pt} \, $}}%
d_{4}x^{\prime })  \notag \\
&=&p^{\mu ^{\prime }\nu ^{\prime }\alpha ^{\prime }\beta ^{\prime }}x_{,\mu
^{\prime }}^{\mu }x_{,\nu ^{\prime }}^{\nu }x_{,\alpha ^{\prime }}^{\alpha
}x_{,\beta ^{\prime }}^{\beta }\det (x_{,\lambda }^{\lambda ^{\prime
}})\left( \mathrm{d}z_{\mu \nu \alpha }-z_{\mu \nu \alpha \gamma }\mathrm{d}%
x^{\gamma }\right) \wedge \partial _{\beta }%
{\mbox{$ \rule {5pt} {.5pt}\rule {.5pt} {6pt} \, $}}%
d_{4}x  \label{Zeta-1} \\
&&+p^{\mu ^{\prime }\nu ^{\prime }\alpha ^{\prime }\beta ^{\prime }}\left(
x_{,\mu ^{\prime }\alpha ^{\prime }}^{\mu }x_{,\nu ^{\prime }}^{\nu
}+x_{,\mu ^{\prime }}^{\mu }x_{,\nu ^{\prime }\alpha ^{\prime }}^{\nu
}\right) x_{,\beta ^{\prime }}^{\beta }\det (x_{,\lambda }^{\lambda ^{\prime
}})\left( \mathrm{d}y_{\mu \nu }-z_{\mu \nu \gamma }\mathrm{d}x^{\gamma
}\right) \wedge \partial _{\beta }%
{\mbox{$ \rule {5pt} {.5pt}\rule {.5pt} {6pt} \, $}}%
d_{4}x.  \notag
\end{eqnarray}%
So that we have derived the first term in the boundary form (\ref{Zeta}) in
terms of the primed coordinates.

As a second step, we write the one-form $\mathrm{d}y_{\mu \nu }-z_{\mu \nu
\beta }\mathrm{d}x^{\beta }$ in terms of the primed coordinates. By
substituting the transformations of $y_{\mu ^{\prime }\nu ^{\prime }}$ in (%
\ref{gN}) and $z_{\mu ^{\prime }\nu ^{\prime }\beta ^{\prime }}$ in (\ref%
{prolonged-g-1}) into this one-form, we have 
\begin{align*}
\mathrm{d}y_{\mu ^{\prime }\nu ^{\prime }}-z_{\mu ^{\prime }\nu ^{\prime
}\beta ^{\prime }}\mathrm{d}x^{\beta ^{\prime }} &=\mathrm{d}(y_{\mu \nu
}x_{,\mu ^{\prime }}^{\mu }x_{,\nu ^{\prime }}^{\nu })-(z_{\mu \nu \beta
}x_{,\mu ^{\prime }}^{\mu }x_{,\nu ^{\prime }}^{\nu }x_{,\beta ^{\prime
}}^{\beta }+y_{\mu \nu }x_{,\mu ^{\prime }\beta ^{\prime }}^{\mu }x_{,\nu
^{\prime }}^{\nu }+y_{\mu \nu }x_{,\mu ^{\prime }}^{\mu }x_{,\nu ^{\prime
}\beta ^{\prime }}^{\nu })\mathrm{d}x^{\beta ^{\prime }} \\
&=x_{,\mu ^{\prime }}^{\mu }x_{,\nu ^{\prime }}^{\nu }\mathrm{d}y_{\mu \nu
}+y_{\mu \nu }x_{,\nu ^{\prime }}^{\nu }\mathrm{d}x_{,\mu ^{\prime }}^{\mu
}+y_{\mu \nu }x_{,\mu ^{\prime }}^{\mu }\mathrm{d}x_{,\nu ^{\prime }}^{\nu }
\\
&-z_{\mu \nu \beta }x_{,\mu ^{\prime }}^{\mu }x_{,\nu ^{\prime }}^{\nu }%
\mathrm{d}x^{\beta }-y_{\mu \nu }x_{,\nu ^{\prime }}^{\nu }\mathrm{d}x_{,\mu
^{\prime }}^{\mu }-y_{\mu \nu }x_{,\nu ^{\prime }}^{\nu }\mathrm{d}x_{,\mu
^{\prime }}^{\mu } \\
&=x_{,\mu ^{\prime }}^{\mu }x_{,\nu ^{\prime }}^{\nu }\left( \mathrm{d}%
y_{\mu \nu }-z_{\mu \nu \beta }\mathrm{d}x^{\beta }\right) .
\end{align*}%
See that the second and the third terms in the second line cancels with the
second and the third terms in the third line, respectively. We take the
exterior product of $\mathrm{d}y_{\mu ^{\prime }\nu ^{\prime }}-z_{\mu
^{\prime }\nu ^{\prime }\beta ^{\prime }}\mathrm{d}x^{\beta ^{\prime }}$ and 
$\partial _{\beta ^{\prime }}%
{\mbox{$ \rule {5pt} {.5pt}\rule {.5pt} {6pt} \, $}}%
d_{4}x^{\prime }$, then multiply this four form by $p^{\mu ^{\prime }\nu
^{\prime }\beta ^{\prime }}$. We obtain the following expression for $\Xi
_{2}$ in terms of primed coordinates, 
\begin{eqnarray}
\Xi _{2}^{\prime } &=&p^{\mu ^{\prime }\nu ^{\prime }\beta ^{\prime
}}(dy_{\mu ^{\prime }\nu ^{\prime }}-z_{\mu ^{\prime }\nu ^{\prime }\beta
^{\prime }}dx^{\beta ^{\prime }})\wedge (\partial _{\beta ^{\prime }}%
{\mbox{$ \rule {5pt} {.5pt}\rule {.5pt} {6pt} \, $}}%
d_{4}x)  \notag \\
&=&p^{\mu ^{\prime }\nu ^{\prime }\beta ^{\prime }}x_{,\mu ^{\prime }}^{\mu
}x_{,\nu ^{\prime }}^{\nu }\left( \mathrm{d}y_{\mu \nu }-z_{\mu \nu \gamma }%
\mathrm{d}x^{\gamma }\right) x_{,\beta ^{\prime }}^{\beta }\det \left(
x_{,\lambda }^{\lambda ^{\prime }}\right) \wedge \partial _{\beta }%
{\mbox{$ \rule {5pt} {.5pt}\rule {.5pt} {6pt} \, $}}%
d_{4}x.  \label{Zeta-2}
\end{eqnarray}

The sum of the four-forms in (\ref{Zeta-1}) and (\ref{Zeta-2}) is the
boundary form in primed coordinates. Explicitly we have that%
\begin{eqnarray*}
\Xi ^{\prime } &=&\Xi _{1}^{\prime }+\Xi _{2}^{\prime }=p^{\mu ^{\prime }\nu
^{\prime }\alpha ^{\prime }\beta ^{\prime }}x_{,\mu ^{\prime }}^{\mu
}x_{,\nu ^{\prime }}^{\nu }x_{,\alpha ^{\prime }}^{\alpha }x_{,\beta
^{\prime }}^{\beta }\det \left( x_{,\lambda }^{\lambda ^{\prime }}\right)
\left( \mathrm{d}z_{\mu \nu \alpha }-z_{\mu \nu \alpha \gamma }\mathrm{d}%
x^{\gamma }\right) \wedge \partial _{\beta }%
{\mbox{$ \rule {5pt} {.5pt}\rule {.5pt} {6pt} \, $}}%
d_{4}x \\
&&+p^{\mu ^{\prime }\nu ^{\prime }\beta ^{\prime }}x_{,\mu ^{\prime }}^{\mu
}x_{,\nu ^{\prime }}^{\nu }x_{,\beta ^{\prime }}^{\beta }\det (x_{,\lambda
}^{\lambda ^{\prime }})\left( \mathrm{d}y_{\mu \nu }-z_{\mu \nu \beta }%
\mathrm{d}x^{\beta }\right) \wedge \partial _{\beta }%
{\mbox{$ \rule {5pt} {.5pt}\rule {.5pt} {6pt} \, $}}%
d_{4}x \\
&&+p^{\mu ^{\prime }\nu ^{\prime }\alpha ^{\prime }\beta ^{\prime }}\left(
x_{,\mu ^{\prime }\alpha ^{\prime }}^{\mu }x_{,\nu ^{\prime }}^{\nu
}x_{,\beta ^{\prime }}^{\beta }+x_{,\mu ^{\prime }}^{\mu }x_{,\nu ^{\prime
}\alpha ^{\prime }}^{\nu }x_{,\beta ^{\prime }}^{\beta }\right) \det
(x_{,\lambda }^{\lambda ^{\prime }})\left( \mathrm{d}y_{\mu \nu }-z_{\mu \nu
\beta }\mathrm{d}x^{\beta }\right) \wedge \partial _{\beta }%
{\mbox{$ \rule {5pt} {.5pt}\rule {.5pt} {6pt} \, $}}%
d_{4}x.
\end{eqnarray*}%
Since $\Xi $ is a 4-form, its expression $\Xi ^{\prime }$ in primed
coordinates gives the same form as the expression in the original
coordinates. That is $\Xi =\Xi ^{\prime },$ which implies that 
\begin{eqnarray*}
p^{\mu \nu \alpha } &=&\left( p^{\mu ^{\prime }\nu ^{\prime }\alpha ^{\prime
}}x_{,\mu ^{\prime }}^{\mu }y_{,\nu ^{\prime }}^{\nu }x_{,\alpha ^{\prime
}}^{\alpha }+p^{\mu ^{\prime }\nu ^{\prime }\alpha ^{\prime }\beta ^{\prime
}}x_{,\beta ^{\prime }}^{\beta }x_{,\alpha ^{\prime }}^{\alpha }\left(
x_{,\mu ^{\prime }\beta }^{\mu }x_{,\nu ^{\prime }}^{\nu }+x_{,\mu ^{\prime
}}^{\mu }x_{,\nu ^{\prime }\beta }^{\nu }\right) \right) \det (x_{,\lambda
}^{\lambda ^{\prime }}), \\
p^{\mu \nu \alpha \beta } &=&p^{\mu ^{\prime }\nu ^{\prime }\alpha ^{\prime
}\beta ^{\prime }}x_{,\mu ^{\prime }}^{\mu }x_{,\nu ^{\prime }}^{\nu
}x_{,\alpha ^{\prime }}^{\alpha }x_{,\beta ^{\prime }}^{\beta }\det
(x_{,\lambda }^{\lambda ^{\prime }}).
\end{eqnarray*}%
This completes the proof of Lemma \ref{Lemma-4.2}$.$
\end{proof}

\smallskip

In Lemma \ref{Lemma-4.2} we showed that arbitrary smooth functions $p^{\mu
\nu \alpha \beta }$ and $p^{\mu \nu \alpha }$ on $J^{3}(M,N)$ define a
4-form 
\begin{equation*}
p^{\mu \nu \alpha \beta }(\mathrm{d}z_{\mu \nu \alpha }-z_{\mu \nu \alpha
\gamma }\mathrm{d}x^{\gamma })\wedge (\partial _{\beta }%
{\mbox{$ \rule {5pt} {.5pt}\rule {.5pt} {6pt} \, $}}%
\mathrm{d}_{4}x)+p^{\mu \nu \alpha }(\mathrm{d}y_{\mu \nu }-z_{\mu \nu \beta
}\mathrm{d}x^{\beta })\wedge (\partial _{\alpha }%
{\mbox{$ \rule {5pt} {.5pt}\rule {.5pt} {6pt} \, $}}%
\mathrm{d}_{4}x)
\end{equation*}%
on $J^{3}(M,N)$ provided that, under coordinate transformations (\ref{gM})
through (\ref{prolonged-g}), they transform according to equations (\ref{p-1}%
) and (\ref{p-2}). In the present case, the coefficients $p^{\mu \nu \alpha
\beta }$ and $p^{\mu \nu \alpha }$ are defined as the Ostrogradski's
momenta, which implies equations (\ref{4.2}) for every section $\sigma $ of $%
\pi :N\rightarrow M$. Note that any function $f$ on $J^{3}(M,N)$ is uniquely
determined by its pull-backs $(j^{3}\sigma )^{\ast }f$ for all sections $%
\sigma $ of $\pi :N\rightarrow M$. Therefore we may write 
\begin{eqnarray}
p^{\mu \nu \alpha \beta } &=&\frac{\partial L}{\partial z_{\mu \nu \alpha
\beta }}\left( x^{\lambda },y_{\mu \nu },z_{\mu \nu \alpha ,}z_{\mu \nu
\alpha \beta }\right) ,  \label{Ostrogradski Momenta} \\
p^{\mu \nu \alpha } &=&\frac{\partial L}{\partial z_{\mu \nu \alpha }}\left(
x^{\lambda },y_{\mu \nu },z_{\mu \nu \alpha ,}z_{\mu \nu \alpha \beta
}\right) -D_{\beta }p^{\mu \nu \alpha \beta },  \notag
\end{eqnarray}%
where $D_{\beta }p^{\mu \nu \alpha \beta }$ is the total divergence given by 
\begin{eqnarray}
D_{\beta }p^{\mu \nu \alpha \beta } &=&D_{\beta }\left[ \frac{\partial L}{%
\partial z_{\mu \nu \alpha \beta }}\left( x^{\lambda },y_{\mu \nu },z_{\mu
\nu \alpha ,}z_{\mu \nu \alpha \beta }\right) \right]
\label{total divergence} \\
&=&\frac{\partial }{\partial x^{\beta }}\frac{\partial L}{\partial z_{\mu
\nu \alpha \beta }}\left( x^{\lambda },y_{\mu \nu },z_{\mu \nu \alpha
,}z_{\mu \nu \alpha \beta }\right) +z_{\sigma \tau \beta }\frac{\partial }{%
\partial y_{\sigma \tau }}\frac{\partial L}{\partial z_{\mu \nu \alpha \beta
}}\left( x^{\lambda },y_{\mu \nu },z_{\mu \nu \alpha ,}z_{\mu \nu \alpha
\beta }\right)  \notag \\
&&+z_{\rho \sigma \tau \beta }\frac{\partial }{\partial z_{\rho \sigma \tau }%
}\frac{\partial L}{\partial z_{\mu \nu \alpha \beta }}\left( x^{\lambda
},y_{\mu \nu },z_{\mu \nu \alpha ,}z_{\mu \nu \alpha \beta }\right)  \notag
\\
&&+z_{\pi \rho \sigma \tau \beta }\frac{\partial }{\partial z_{\pi \rho
\sigma \tau }}\frac{\partial L}{\partial z_{\mu \nu \alpha \beta }}\left(
x^{\lambda },y_{\mu \nu },z_{\mu \nu \alpha ,}z_{\mu \nu \alpha \beta
}\right) .  \notag
\end{eqnarray}%
In order to simplify computations, we use the notation $P^{\mu \nu \alpha
\beta }(x^{\gamma })=(j^{3}\sigma )^{\ast }p^{\mu \nu \alpha \beta
}(x^{\gamma })$ and $P^{\mu \nu \alpha }(x^{\gamma })=(j^{3}\sigma )^{\ast
}p^{\mu \nu \alpha }(x^{\gamma })$ introduced in equation (\ref{3.5x}). With
this notation, 
\begin{equation}
P^{\mu \nu \alpha }(x^{\lambda })=(j^{2}\sigma )^{\ast }\left( \frac{%
\partial L}{\partial z_{\mu \nu \alpha }}\right) (x^{\lambda })-P_{,\beta
}^{\mu \nu \alpha \beta }(x^{\lambda }),  \label{divergence}
\end{equation}%
where $P_{,\beta }^{\mu \nu \alpha \beta }(x^{\lambda })=\frac{\partial }{%
\partial x^{\beta }}P^{\mu \nu \alpha \beta }(x^{\lambda }).$

\begin{lemma}
\label{Lemma-4.3} If $Ld_{4}x$ is a second order Lagrangian form on $%
J^{2}(M,N)$, invariant under the coordinate transformations (\ref{gM})
through (\ref{prolonged-g}), Ostrogradski's momenta given by equation (\ref%
{Ostrogradski Momenta}) satisfy the transformation rules (\ref{p-1}) and (%
\ref{p-2}).
\end{lemma}

\begin{proof}
We start with the second momentum $p^{\mu \nu \alpha \beta }$ and obtain the
transformation law (\ref{p-2}) as follows, 
\begin{eqnarray}
p^{\mu \nu \alpha \beta } &=&\frac{\partial }{\partial z_{\mu \nu \alpha
\beta }}L\left( x^{\mu },y_{\mu \nu },z_{\mu \nu \alpha ,}z_{\mu \nu \alpha
\beta }\right)  \label{second momentum} \\
&=&\frac{\partial }{\partial z_{\mu ^{\prime }\nu ^{\prime }\alpha ^{\prime
}\beta ^{\prime }}}L\left( x^{\mu ^{\prime }},y_{\mu ^{\prime }\nu ^{\prime
}},z_{\mu ^{\prime }\nu ^{\prime }\alpha ^{\prime },}z_{\mu ^{\prime }\nu
^{\prime }\alpha ^{\prime }\beta ^{\prime }}\right) \det (x_{,\lambda
}^{\lambda ^{\prime }})\frac{\partial z_{\mu ^{\prime }\nu ^{\prime }\alpha
^{\prime }\beta ^{\prime }}}{\partial z_{\mu \nu \alpha \beta }}  \notag \\
&=&p^{\mu ^{\prime }\nu ^{\prime }\alpha ^{\prime }\beta ^{\prime }}\det
(x_{,\lambda }^{\lambda ^{\prime }})x_{,\mu ^{\prime }}^{\mu }x_{,\nu
^{\prime }}^{\nu }x_{,\alpha ^{\prime }}^{\alpha }x_{,\beta ^{\prime
}}^{\beta },  \notag
\end{eqnarray}%
where we have used the chain rule, the prolonged coordinate transformation
for $z_{\mu ^{\prime }\nu ^{\prime }\alpha ^{\prime }\beta ^{\prime }}$
given in (\ref{prolonged-g}), and equation (\ref{inv-Lag}).

In order to show that the first momentum $p^{\mu \nu \alpha }$ satisfies
transformation law (\ref{p-1}), start first with the term $\partial
L/\partial z_{\mu \nu \alpha }.$ As in equation (\ref{second momentum}), 
\begin{eqnarray}
&&\frac{\partial L}{\partial z_{\mu \nu \alpha }}=\frac{\partial z_{\mu
^{\prime }\nu ^{\prime }\alpha ^{\prime }}}{\partial z_{\mu \nu \alpha }}%
\frac{\partial L}{\partial z_{\mu ^{\prime }\nu ^{\prime }\alpha ^{\prime }}}%
\det (x_{,\lambda }^{\lambda ^{\prime }})+\frac{\partial z_{\mu ^{\prime
}\nu ^{\prime }\alpha ^{\prime }\beta ^{\prime }}}{\partial z_{\mu \nu
\alpha }}\frac{\partial L}{\partial z_{\mu ^{\prime }\nu ^{\prime }\alpha
^{\prime }\beta ^{\prime }}}\det (x_{,\lambda }^{\lambda ^{\prime }})  \notag
\\
&=&x_{,\mu ^{\prime }}^{\mu }x_{,\nu ^{\prime }}^{\nu }x_{,\alpha ^{\prime
}}^{\alpha }\frac{\partial L}{\partial z_{\mu ^{\prime }\nu ^{\prime }\alpha
^{\prime }}}\det (x_{,\lambda }^{\lambda ^{\prime }})  \notag \\
&&+\left( (x_{,\mu ^{\prime }}^{\mu }x_{,\nu ^{\prime }}^{\nu }x_{,\alpha
^{\prime }}^{\alpha })_{,\beta ^{\prime }}+x_{,\beta ^{\prime }}^{\alpha
}x_{,\mu ^{\prime }\alpha ^{\prime }}^{\mu }x_{,\nu ^{\prime }}^{\nu
}+x_{,\beta ^{\prime }}^{\alpha }x_{,\mu ^{\prime }}^{\mu }x_{,\nu ^{\prime
}\alpha ^{\prime }}^{\nu }\right) \frac{\partial L}{\partial z_{\mu ^{\prime
}\nu ^{\prime }\alpha ^{\prime }\beta ^{\prime }}}\det (x_{,\lambda
}^{\lambda ^{\prime }}).  \label{Ost-Mom-1-1}
\end{eqnarray}%
In order to compute the divergence term, we work with pull-backs $P^{\mu \nu
\alpha }=(j^{3}\sigma )^{\ast }p^{\mu \nu \alpha }$ and $P^{\mu \nu \alpha
\beta }=(j^{3}\sigma )^{\ast }p^{\mu \nu \alpha \beta }$, which allows
replacing total derivative by partial derivative, see equation (\ref%
{divergence}). We obtain 
\begin{eqnarray}
j^{3}\sigma ^{\ast }(D_{\beta }p^{\mu \nu \alpha \beta }) &=&\left(
j^{3}\sigma ^{\ast }p^{\mu \nu \alpha \beta }\right) _{,\beta }=P_{,\beta
}^{\mu \nu \alpha \beta }=\left( P^{\mu ^{\prime }\nu ^{\prime }\alpha
^{\prime }\beta ^{\prime }}\det (x_{,\lambda }^{\lambda ^{\prime }})x_{,\mu
^{\prime }}^{\mu }x_{,\nu ^{\prime }}^{\nu }x_{,\alpha ^{\prime }}^{\alpha
}x_{,\beta ^{\prime }}^{\beta }\right) _{,\beta }  \notag \\
&=&\left( P^{\mu ^{\prime }\nu ^{\prime }\alpha ^{\prime }\beta ^{\prime
}}\det (x_{,\lambda }^{\lambda ^{\prime }})x_{,\mu ^{\prime }}^{\mu }x_{,\nu
^{\prime }}^{\nu }x_{,\alpha ^{\prime }}^{\alpha }x_{,\beta ^{\prime
}}^{\beta }\right) _{,\gamma ^{\prime }}x_{,\beta }^{\gamma ^{\prime }} 
\notag \\
&=&P_{,\beta ^{\prime }}^{\mu ^{\prime }\nu ^{\prime }\alpha ^{\prime }\beta
^{\prime }}\det (x_{,\lambda }^{\lambda ^{\prime }})x_{,\mu ^{\prime }}^{\mu
}x_{,\nu ^{\prime }}^{\nu }x_{,\alpha ^{\prime }}^{\alpha }+P^{\mu ^{\prime
}\nu ^{\prime }\alpha ^{\prime }\beta ^{\prime }}\det (x_{,\lambda
}^{\lambda ^{\prime }})_{,\beta ^{\prime }}(x_{,\mu ^{\prime }}^{\mu
}x_{,\nu ^{\prime }}^{\nu }x_{,\alpha ^{\prime }}^{\alpha })\notag
\\
&&+P^{\mu ^{\prime }\nu ^{\prime }\alpha ^{\prime }\beta ^{\prime }}\det
(x_{,\lambda }^{\lambda ^{\prime }})(x_{,\mu ^{\prime }}^{\mu }x_{,\nu
^{\prime }}^{\nu }x_{,\alpha ^{\prime }}^{\alpha })_{,\beta ^{\prime
}} \notag \\ &&+P^{\mu ^{\prime }\nu ^{\prime }\alpha ^{\prime }\beta ^{\prime }}\det
(x_{,\lambda }^{\lambda ^{\prime }})x_{,\mu ^{\prime }}^{\mu }x_{,\nu
^{\prime }}^{\nu }x_{,\alpha ^{\prime }}^{\alpha }x_{,\beta ^{\prime }\gamma
^{\prime }}^{\beta }x_{,\beta }^{\gamma ^{\prime }}.  \label{Ost-Mom-1-2} 
\end{eqnarray}%
Note that 
\begin{equation}
\det (x_{,\lambda }^{\lambda ^{\prime }})_{,\beta ^{\prime }}=\frac{\partial
\det (x_{,\lambda }^{\lambda ^{\prime }})}{\partial x_{,\beta }^{\gamma
^{\prime }}}\frac{\partial x_{,\beta }^{\gamma ^{\prime }}}{\partial
x^{\beta ^{\prime }}}=\det (x_{,\lambda }^{\lambda ^{\prime }})x_{,\gamma
^{\prime }}^{\beta }\frac{\partial x_{,\beta }^{\gamma ^{\prime }}}{\partial
x^{\beta ^{\prime }}},  \label{detprop}
\end{equation}%
and%
\begin{equation*}
\frac{\partial }{\partial x^{\beta ^{\prime }}}\left( x_{,\beta }^{\gamma
^{\prime }}x_{,\gamma ^{\prime }}^{\beta }\right) =\frac{\partial }{\partial
x^{\beta ^{\prime }}}\delta _{\gamma ^{\prime }}^{\gamma ^{\prime }}=0,
\end{equation*}%
which implies 
\begin{equation}
x_{,\gamma ^{\prime }}^{\beta }\frac{\partial x_{,\beta }^{\gamma ^{\prime }}%
}{\partial x^{\beta ^{\prime }}}=-x_{,\beta }^{\gamma ^{\prime }}\frac{%
\partial x_{,\gamma ^{\prime }}^{\beta }}{\partial x^{\beta ^{\prime }}}.
\label{ident}
\end{equation}%
Therefore, the second term on the right hand side of equation (\ref%
{Ost-Mom-1-2}) reads%
\begin{eqnarray*}
P^{\mu ^{\prime }\nu ^{\prime }\alpha ^{\prime }\beta ^{\prime }}\det
(x_{,\lambda }^{\lambda ^{\prime }})_{,\beta ^{\prime }}(x_{,\mu ^{\prime
}}^{\mu }x_{,\nu ^{\prime }}^{\nu }x_{,\alpha ^{\prime }}^{\alpha })
&=&P^{\mu ^{\prime }\nu ^{\prime }\alpha ^{\prime }\beta ^{\prime }}\det
(x_{,\lambda }^{\lambda ^{\prime }})x_{,\gamma ^{\prime }}^{\beta }\frac{%
\partial x_{,\beta }^{\gamma ^{\prime }}}{\partial x^{\beta ^{\prime }}}%
x_{,\mu ^{\prime }}^{\mu }x_{,\nu ^{\prime }}^{\nu }x_{,\alpha ^{\prime
}}^{\alpha } \\
&=&-P^{\mu ^{\prime }\nu ^{\prime }\alpha ^{\prime }\beta ^{\prime }}\det
(x_{,\lambda }^{\lambda ^{\prime }})x_{,\beta }^{\gamma ^{\prime
}}x_{,\gamma ^{\prime }\beta ^{\prime }}^{\beta }x_{,\mu ^{\prime }}^{\mu
}x_{,\nu ^{\prime }}^{\nu }x_{,\alpha ^{\prime }}^{\alpha },
\end{eqnarray*}%
In the light of this, we can rewrite $P_{,\beta }^{\mu \nu \alpha \beta }$
in (\ref{Ost-Mom-1-2}) as%
\begin{eqnarray}
P_{,\beta }^{\mu \nu \alpha \beta } &=&P_{,\beta ^{\prime }}^{\mu ^{\prime
}\nu ^{\prime }\alpha ^{\prime }\beta ^{\prime }}\det (x_{,\lambda
}^{\lambda ^{\prime }})x_{,\mu ^{\prime }}^{\mu }x_{,\nu ^{\prime }}^{\nu
}x_{,\alpha ^{\prime }}^{\alpha }-P^{\mu ^{\prime }\nu ^{\prime }\alpha
^{\prime }\beta ^{\prime }}\det (x_{,\lambda }^{\lambda ^{\prime
}})x_{,\beta }^{\gamma ^{\prime }}x_{,\gamma ^{\prime }\beta ^{\prime
}}^{\beta }x_{,\mu ^{\prime }}^{\mu }x_{,\nu ^{\prime }}^{\nu }x_{,\alpha
^{\prime }}^{\alpha }  \notag \\
&&+P^{\mu ^{\prime }\nu ^{\prime }\alpha ^{\prime }\beta ^{\prime }}\det
(x_{,\lambda }^{\lambda ^{\prime }})(x_{,\mu ^{\prime }}^{\mu }x_{,\nu
^{\prime }}^{\nu }x_{,\alpha ^{\prime }}^{\alpha })_{,\beta ^{\prime
}}+P^{\mu ^{\prime }\nu ^{\prime }\alpha ^{\prime }\beta ^{\prime }}\det
(x_{,\lambda }^{\lambda ^{\prime }})x_{,\mu ^{\prime }}^{\mu }x_{,\nu
^{\prime }}^{\nu }x_{,\alpha ^{\prime }}^{\alpha }x_{,\beta ^{\prime }\gamma
^{\prime }}^{\beta }x_{,\beta }^{\gamma ^{\prime }}  \notag \\
&=&P_{,\beta ^{\prime }}^{\mu ^{\prime }\nu ^{\prime }\alpha ^{\prime }\beta
^{\prime }}\det (x_{,\lambda }^{\lambda ^{\prime }})x_{,\mu ^{\prime }}^{\mu
}x_{,\nu ^{\prime }}^{\nu }x_{,\alpha ^{\prime }}^{\alpha }+P^{\mu ^{\prime
}\nu ^{\prime }\alpha ^{\prime }\beta ^{\prime }}\det (x_{,\lambda
}^{\lambda ^{\prime }})(x_{,\mu ^{\prime }}^{\mu }x_{,\nu ^{\prime }}^{\nu
}x_{,\alpha ^{\prime }}^{\alpha })_{,\beta ^{\prime }}  \notag \\
&&+P^{\mu ^{\prime }\nu ^{\prime }\alpha ^{\prime }\beta ^{\prime }}\det
(x_{,\lambda }^{\lambda ^{\prime }})x_{,\mu ^{\prime }}^{\mu }x_{,\nu
^{\prime }}^{\nu }x_{,\alpha ^{\prime }}^{\alpha }(x_{,\beta ^{\prime
}\gamma ^{\prime }}^{\beta }x_{,\beta }^{\gamma ^{\prime }}-x_{,\beta
}^{\gamma ^{\prime }}x_{,\gamma ^{\prime }\beta ^{\prime }}^{\beta })  \notag
\\
&=&P_{,\beta ^{\prime }}^{\mu ^{\prime }\nu ^{\prime }\alpha ^{\prime }\beta
^{\prime }}\det (x_{,\lambda }^{\lambda ^{\prime }})x_{,\mu ^{\prime }}^{\mu
}x_{,\nu ^{\prime }}^{\nu }x_{,\alpha ^{\prime }}^{\alpha }+P^{\mu ^{\prime
}\nu ^{\prime }\alpha ^{\prime }\beta ^{\prime }}\det (x_{,\lambda
}^{\lambda ^{\prime }})(x_{,\mu ^{\prime }}^{\mu }x_{,\nu ^{\prime }}^{\nu
}x_{,\alpha ^{\prime }}^{\alpha })_{,\beta ^{\prime }}.
\label{Ost-Mom-1-2-2}
\end{eqnarray}%
In order to arrive at the coordinate transformation for Ostrogradski's
momentum $P^{\mu \nu \alpha },$ we simply take the difference of $%
(j^{2}\sigma )^{\ast }\left( \partial L/\partial z_{\mu \nu \alpha }\right) $
in (\ref{Ost-Mom-1-1}) and $P_{,\beta }^{\mu \nu \alpha \beta }$ in (\ref%
{Ost-Mom-1-2-2}). So that, 
\begin{eqnarray}
P^{\mu \nu \alpha } &=&(j^{2}\sigma )^{\ast }\left( \frac{\partial L}{%
\partial z_{\mu \nu \alpha }}\right) \left( x^{\mu },y_{\mu \nu },z_{\mu \nu
\alpha ,}z_{\mu \nu \alpha \beta }\right) -P_{,\beta }^{\mu \nu \alpha \beta
} 
 \notag \\
&=&x_{,\mu ^{\prime }}^{\mu }x_{,\nu ^{\prime }}^{\nu }x_{,\alpha ^{\prime
}}^{\alpha }(j^{2}\sigma )^{\ast }\left( \frac{\partial L}{\partial z_{\mu
^{\prime }\nu ^{\prime }\alpha ^{\prime }}}\right) \det (x_{,\lambda
}^{\lambda ^{\prime }}) \notag \\
&& +(x_{,\mu ^{\prime }}^{\mu }x_{,\nu ^{\prime }}^{\nu
}x_{,\alpha ^{\prime }}^{\alpha })_{,\beta ^{\prime }}(j^{2}\sigma )^{\ast
}\left( \frac{\partial L}{\partial z_{\mu ^{\prime }\nu ^{\prime }\alpha
^{\prime }\beta ^{\prime }}}\right) \det (x_{,\lambda }^{\lambda ^{\prime
}})
 \notag \\
&&+\left( x_{,\beta ^{\prime }}^{\alpha }x_{,\mu ^{\prime }\alpha ^{\prime
}}^{\mu }x_{,\nu ^{\prime }}^{\nu }+x_{,\beta ^{\prime }}^{\alpha }x_{,\mu
^{\prime }}^{\mu }x_{,\nu ^{\prime }\alpha ^{\prime }}^{\nu }\right)
(j^{2}\sigma )^{\ast }\left( \frac{\partial L}{\partial z_{\mu ^{\prime }\nu
^{\prime }\alpha ^{\prime }\beta ^{\prime }}}\right) \det (x_{,\lambda
}^{\lambda ^{\prime }})  \notag \\
&&-P_{,\beta ^{\prime }}^{\mu ^{\prime }\nu ^{\prime }\alpha ^{\prime }\beta
^{\prime }}\det (x_{,\lambda }^{\lambda ^{\prime }})x_{,\mu ^{\prime }}^{\mu
}x_{,\nu ^{\prime }}^{\nu }x_{,\alpha ^{\prime }}^{\alpha }-P^{\mu ^{\prime
}\nu ^{\prime }\alpha ^{\prime }\beta ^{\prime }}\det (x_{,\lambda
}^{\lambda ^{\prime }})(x_{,\mu ^{\prime }}^{\mu }x_{,\nu ^{\prime }}^{\nu
}x_{,\alpha ^{\prime }}^{\alpha })_{,\beta ^{\prime }}  \notag \\
&=&\left( x_{,\mu ^{\prime }}^{\mu }x_{,\nu ^{\prime }}^{\nu }x_{,\alpha
^{\prime }}^{\alpha }\right) \det (x_{,\lambda }^{\lambda ^{\prime }})\left(
(j^{2}\sigma )^{\ast }\left( \frac{\partial L}{\partial z_{\mu ^{\prime }\nu
^{\prime }\alpha ^{\prime }}}\right) -P_{,\beta ^{\prime }}^{\mu ^{\prime
}\nu ^{\prime }\alpha ^{\prime }\beta ^{\prime }}\det (x_{,\lambda ^{\prime
}}^{\lambda })\right)  \notag \\
&&+\det (x_{,\lambda }^{\lambda ^{\prime }})(x_{,\mu ^{\prime }}^{\mu
}x_{,\nu ^{\prime }}^{\nu }x_{,\alpha ^{\prime }}^{\alpha })_{,\beta
^{\prime }}\left( (j^{2}\sigma )^{\ast }\left( \frac{\partial L}{\partial
z_{\mu ^{\prime }\nu ^{\prime }\alpha ^{\prime }\beta ^{\prime }}}\right)
-P^{\mu ^{\prime }\nu ^{\prime }\alpha ^{\prime }\beta ^{\prime }}\right) 
\notag \\
&&+\left( x_{,\beta ^{\prime }}^{\alpha }x_{,\mu ^{\prime }\alpha ^{\prime
}}^{\mu }x_{,\nu ^{\prime }}^{\nu }+x_{,\beta ^{\prime }}^{\alpha }x_{,\mu
^{\prime }}^{\mu }x_{,\nu ^{\prime }\alpha ^{\prime }}^{\nu }\right)
(j^{2}\sigma )^{\ast }\left( \frac{\partial L}{\partial z_{\mu ^{\prime }\nu
^{\prime }\alpha ^{\prime }\beta ^{\prime }}}\right) \det (x_{,\lambda
}^{\lambda ^{\prime }})  \label{Ost-Mom-1-2-3} \\
&=&P^{\mu ^{\prime }\nu ^{\prime }\alpha ^{\prime }}x_{,\mu ^{\prime }}^{\mu
}x_{,\nu ^{\prime }}^{\nu }x_{,\alpha ^{\prime }}^{\alpha }\det (x_{,\lambda
}^{\lambda ^{\prime }})+\left( x_{,\beta ^{\prime }}^{\alpha }x_{,\mu
^{\prime }\alpha ^{\prime }}^{\mu }x_{,\nu ^{\prime }}^{\nu }+x_{,\beta
^{\prime }}^{\alpha }x_{,\mu ^{\prime }}^{\mu }x_{,\nu ^{\prime }\alpha
^{\prime }}^{\nu }\right) P^{\mu ^{\prime }\nu ^{\prime }\alpha ^{\prime
}\beta ^{\prime }}\det (x_{,\lambda }^{\lambda ^{\prime }}).  \notag
\end{eqnarray}%
Here, we used notation (\ref{3.5x}) in the primed coordinates, which yields 
\begin{eqnarray*}
P^{\mu ^{\prime }\nu ^{\prime }\alpha ^{\prime }\beta ^{\prime }}
&=&(j^{3}\sigma )^{\ast }p^{\mu ^{\prime }\nu ^{\prime }\alpha ^{\prime
}\beta \prime }=(j^{2}\sigma )^{\ast }\left( \frac{\partial L}{\partial
z_{\mu ^{\prime }\nu ^{\prime }\alpha ^{\prime }\beta ^{\prime }}}\right) ,
\\
P^{\mu ^{\prime }\nu ^{\prime }\alpha ^{\prime }} &=&(j^{3}\sigma )^{\ast
}p^{\mu ^{\prime }\nu ^{\prime }\alpha ^{\prime }}=(j^{2}\sigma )^{\ast }%
\frac{\partial L}{\partial z_{\mu ^{\prime }\nu ^{\prime }\alpha ^{\prime }}}%
-(j^{3}\sigma )^{\ast }\left( D_{\beta ^{\prime }}\frac{\partial L}{\partial
z_{\mu ^{\prime }\nu ^{\prime }\alpha ^{\prime }\beta ^{\prime }}}\right) \\
&=&(j^{2}\sigma )^{\ast }\left( \frac{\partial L}{\partial z_{\mu ^{\prime
}\nu ^{\prime }\alpha ^{\prime }}}\right) -P_{,\beta ^{\prime }}^{\mu
^{\prime }\nu ^{\prime }\alpha ^{\prime }\beta ^{\prime }}.
\end{eqnarray*}%
Equation (\ref{Ost-Mom-1-2-3}) may be rewritten as 
\begin{equation*}
P^{\mu \nu \alpha }=P^{\mu ^{\prime }\nu ^{\prime }\alpha ^{\prime }}x_{,\mu
^{\prime }}^{\mu }x_{,\nu ^{\prime }}^{\nu }x_{,\alpha ^{\prime }}^{\alpha
}\det (x_{,\lambda }^{\lambda ^{\prime }})+\left( x_{,\beta ^{\prime
}}^{\alpha }x_{,\mu ^{\prime }\alpha ^{\prime }}^{\mu }x_{,\nu ^{\prime
}}^{\nu }+x_{,\beta ^{\prime }}^{\alpha }x_{,\mu ^{\prime }}^{\mu }x_{,\nu
^{\prime }\alpha ^{\prime }}^{\nu }\right) P^{\mu ^{\prime }\nu ^{\prime
}\alpha ^{\prime }\beta ^{\prime }}\det (x_{,\lambda }^{\lambda ^{\prime }}).
\end{equation*}%
Since this equation is valid for every section $\sigma $ of $\pi
:N\rightarrow M$, it follows that 
\begin{equation*}
p^{\mu \nu \alpha }=p^{\mu ^{\prime }\nu ^{\prime }\alpha ^{\prime }}x_{,\mu
^{\prime }}^{\mu }x_{,\nu ^{\prime }}^{\nu }x_{,\alpha ^{\prime }}^{\alpha
}\det (x_{,\lambda }^{\lambda ^{\prime }})+\left( x_{,\beta ^{\prime
}}^{\alpha }x_{,\mu ^{\prime }\alpha ^{\prime }}^{\mu }x_{,\nu ^{\prime
}}^{\nu }+x_{,\beta ^{\prime }}^{\alpha }x_{,\mu ^{\prime }}^{\mu }x_{,\nu
^{\prime }\alpha ^{\prime }}^{\nu }\right) p^{\mu ^{\prime }\nu ^{\prime
}\alpha ^{\prime }\beta ^{\prime }}\det (x_{,\lambda }^{\lambda ^{\prime }}),
\end{equation*}%
where $p^{\mu \nu \alpha }$ and $p^{\mu ^{\prime }\nu ^{\prime }\alpha
^{\prime }\beta ^{\prime }}$ are Ostrogradski's momenta corresponding to the
Lagrangian form $Ld_{4}x$, (\ref{Ostrogradski Momenta}). This completes
proof of Lemma \ref{Lemma-4.3}. \smallskip
\end{proof}

It follows from Lemma \ref{Lemma-4.2} and Lemma \ref{Lemma-4.3} that for an
invariant Lagrangian form $Ld_{4}x$, the corresponding boundary form has the
same expression in the class of coordinate system on $M$, which differ by
orientation preserving transformations. Therefore, the boundary form $\Xi $
is globally defined and is given by a natural differential operator applied
to the to the Lagrangian form $Ld_{4}x$. Since $\Theta =\pi _{2}^{3\ast }L%
\mathrm{d}_{4}x+\Xi $, it follows that the De Donder form $\Theta $ is
globally defined and is also given by a natural differential operator
applied to the to the Lagrangian form $Ld_{4}x$. This completes proof of
Theorem 2.

\subsection{Proof of Proposition 5}

The outline of the proof is as follows. First, we will write the Hilbert
Lagrangian (\ref{3.9a}) in terms of the metric tensor and its partial
derivatives. Such kind of a local presentation of the Hilbert Lagrangian
will enable us to prove the Lemma \ref{Lemma-5.1} where we shall exhibit the
induced Ostrogradski's momenta. Then, we will be ready for the calculation
of the De Donder form (\ref{3.9b}) in an explicit form.

Recall that the Christoffel symbols of the first kind $\Gamma _{\lambda \mu
\nu }$ and the Christoffel symbols of the second kind $\Gamma _{\mu \nu
}^{\rho }$ are defined and related as%
\begin{equation}
\Gamma _{\mu \nu }^{\rho }=g^{\rho \lambda }\Gamma _{\lambda \mu \nu }=\frac{%
1}{2}g^{\rho \lambda }(g_{\mu \lambda ,\nu }+g_{\nu \lambda ,\mu }-g_{\mu
\nu ,\lambda }),  \label{Christoffel}
\end{equation}%
where $g^{\rho \lambda }$ is the dual of the metric tensor $g_{\rho \lambda
} $ whereas $g_{\mu \lambda ,\nu }$ denotes the partial derivative of $%
g_{\mu \lambda }$ with respect to $x^{\nu }$. It is possible to write the
Christoffel symbols in a pure contravariant form%
\begin{equation}
\Gamma ^{\lambda \mu \nu }=g^{\lambda \alpha }g^{\mu \beta }g^{\nu \gamma
}\Gamma _{\alpha \beta \gamma }.  \label{Gamma-contravariant}
\end{equation}%
For future reference, we define here some symbols by contacting the
Christoffel symbols 
\begin{eqnarray}
\Gamma ^{\lambda } &:&=g^{\mu \nu }\Gamma _{\mu \nu }^{\lambda }=g_{\mu \nu
}\Gamma ^{\lambda \mu \nu },  \label{G-T} \\
\Delta ^{\nu } &:&=g^{\nu \mu }\Delta _{\mu }=g_{\lambda \mu }\Gamma
^{\lambda \nu \mu },  \label{D-T} \\
\Gamma _{\rho } &:&=g_{\rho \lambda }\Gamma ^{\lambda }=g^{\mu \nu }\Gamma
_{\rho \mu \nu },  \label{G-D} \\
\Delta _{\mu } &:&=\Gamma _{\mu \lambda }^{\lambda }=g^{\lambda \nu }\Gamma
_{\lambda \mu \nu }.  \label{D-D}
\end{eqnarray}%
Taking the derivative of the identity $g^{\rho \lambda }g_{\lambda \mu
}=\delta _{\mu }^{\rho }$, one arrives at the relation between $g_{,\gamma
}^{\rho \lambda }$ and $\Gamma _{\lambda \mu \nu }$ as follows 
\begin{equation}
g_{,\gamma }^{\alpha \delta }=-g^{\alpha \mu }g^{\delta \nu }g_{\mu \nu
,\gamma }=-g^{\alpha \mu }g^{\delta \nu }\left( \Gamma _{\mu \nu \gamma
}+\Gamma _{\nu \mu \gamma }\right) ,  \label{simpl-1}
\end{equation}%
whereas the contraction of this yields%
\begin{equation}
g_{,\alpha }^{\alpha \delta }=-g^{\alpha \mu }g^{\delta \nu }\left( \Gamma
_{\mu \nu \alpha }+\Gamma _{\nu \mu \alpha }\right) .  \label{simpl-2}
\end{equation}

Recall also that, the Riemann and the Ricci tensors are 
\begin{eqnarray}
R_{\beta \gamma \delta }^{\alpha } &=&\Gamma _{\beta \delta ,\gamma
}^{\alpha }-\Gamma _{\beta \gamma ,\delta }^{\alpha }+\Gamma _{\mu \gamma
}^{\alpha }\Gamma _{\beta \delta }^{\mu }-\Gamma _{\mu \delta }^{\alpha
}\Gamma _{\beta \gamma }^{\mu }.  \label{Riemann} \\
R_{\beta \delta } &=&R_{\beta \alpha \delta }^{\alpha }=\Gamma _{\beta
\delta ,\alpha }^{\alpha }-\Gamma _{\beta \alpha ,\delta }^{\alpha }+\Gamma
_{\mu \alpha }^{\alpha }\Gamma _{\beta \delta }^{\mu }-\Gamma _{\mu \delta
}^{\alpha }\Gamma _{\beta \alpha }^{\mu },  \label{Ricci}
\end{eqnarray}%
respectively. Here, $\Gamma _{\beta \delta ,\gamma }^{\alpha }$ denotes the
partial derivative of $\Gamma _{\beta \delta }^{\alpha }$ with respect to $%
x^{\gamma }$. In this local representation, the scalar curvature is defined
to be 
\begin{equation}
R=g^{\beta \gamma }R_{\beta \gamma }=g^{\beta \gamma }\left( \Gamma _{\beta
\gamma ,\alpha }^{\alpha }-\Gamma _{\beta \alpha ,\gamma }^{\alpha }+\Gamma
_{\mu \alpha }^{\alpha }\Gamma _{\beta \gamma }^{\mu }-\Gamma _{\mu \gamma
}^{\alpha }\Gamma _{\beta \alpha }^{\mu }\right) .  \label{Scalar}
\end{equation}%
Note that the presentation (\ref{Scalar}) is in terms of the Christoffel
symbols of the second kind. It is possible to write $R$ in terms of the
Christoffel symbols of the first kind and its partial derivative as well.
Simply, by substituting the definition in (\ref{Christoffel}), we compute%
\begin{eqnarray*}
R &=&g^{\beta \gamma }\left( g^{\alpha \delta }\Gamma _{\delta \beta \gamma
}\right) _{,\alpha }-g^{\beta \gamma }\left( g^{\alpha \delta }\Gamma
_{\delta \beta \alpha }\right) _{,\gamma }+g^{\beta \gamma }g^{\alpha \rho
}g^{\mu \sigma }(\Gamma _{\rho \mu \alpha }\Gamma _{\sigma \beta \gamma
}-\Gamma _{\rho \mu \gamma }\Gamma _{\sigma \beta \alpha }) \\
&=&g^{\beta \gamma }g^{\alpha \delta }\left( \Gamma _{\delta \beta \gamma
,\alpha }-\Gamma _{\delta \beta \alpha ,\gamma }\right) +g^{\beta \gamma
}\left( g_{,\alpha }^{\alpha \delta }\Gamma _{\delta \beta \gamma
}-g_{,\gamma }^{\alpha \delta }\Gamma _{\delta \beta \alpha }\right) \\
&&+g^{\beta \gamma }g^{\alpha \rho }g^{\mu \sigma }\left( \Gamma _{\rho \mu
\alpha }\Gamma _{\sigma \beta \gamma }-\Gamma _{\rho \mu \gamma }\Gamma
_{\sigma \beta \alpha }\right) .
\end{eqnarray*}%
Notice that, the symbols $\Gamma _{\rho \mu \alpha }$ contain the first
derivative $g_{\mu \nu ,\lambda }$ of the metric $g_{\mu \nu }$, so that the
partial derivative $\Gamma _{\delta \beta \gamma ,\alpha }$ of the symbols
are containing the second partial derivative $g_{\mu \nu ,\lambda \gamma }$
of $g_{\mu \nu }$. In accordance with this, we understand the scalar
curvature $R$ as the sum of two terms, say $R_{1}$ and $R_{2}$ by putting
all the first order terms that is those involving $\Gamma _{\rho \mu \alpha
} $ into $R_{1}$, and by putting all the second order terms that is those
involving $\Gamma _{\delta \beta \gamma ,\alpha }$ into $R_{2}$, that is $%
R=R_{1}+R_{2}$ and 
\begin{eqnarray*}
R_{1} &=&g^{\beta \gamma }\left( g_{,\alpha }^{\alpha \delta }\Gamma
_{\delta \beta \gamma }-g_{,\gamma }^{\alpha \delta }\Gamma _{\delta \beta
\alpha }\right) +g^{\beta \gamma }g^{\alpha \rho }g^{\mu \sigma }\left(
\Gamma _{\rho \mu \alpha }\Gamma _{\sigma \beta \gamma }-\Gamma _{\rho \mu
\gamma }\Gamma _{\sigma \beta \alpha }\right) \\
R_{2} &=&g^{\beta \gamma }g^{\alpha \delta }\left( \Gamma _{\delta \beta
\gamma ,\alpha }-\Gamma _{\delta \beta \alpha ,\gamma }\right) .
\end{eqnarray*}%
Therefore, we write the Lagrangian as%
\begin{equation}
L=R_{1}\sqrt{-\det g}+R_{2}\sqrt{-\det g},  \label{L2}
\end{equation}%
where $R_{1}$ depends linearly on $g_{\mu \nu ,\alpha }$ linearly while $%
R_{2}$ depends quadratically on $g_{\mu \nu ,\alpha \beta }$. A
simplification is possible for $R_{1}$. See that, 
\begin{eqnarray}
R_{1} &=&g^{\beta \gamma }\left( -g^{\alpha \mu }g^{\delta \nu }\left(
\Gamma _{\mu \nu \alpha }+\Gamma _{\nu \mu \alpha }\right) \Gamma _{\delta
\beta \gamma }+g^{\alpha \mu }g^{\delta \nu }\left( \Gamma _{\mu \nu \gamma
}+\Gamma _{\nu \mu \gamma }\right) \Gamma _{\delta \beta \alpha }\right) 
\notag \\
&&+g^{\beta \gamma }g^{\alpha \rho }g^{\mu \sigma }\left( \Gamma _{\rho \mu
\alpha }\Gamma _{\sigma \beta \gamma }-\Gamma _{\rho \mu \gamma }\Gamma
_{\sigma \beta \alpha }\right)  \notag \\
&=&\Gamma _{\mu \nu \lambda }\Gamma _{\alpha \beta \gamma }\left( -g^{\mu
\lambda }g^{\nu \alpha }g^{\beta \gamma }-g^{\mu \alpha }g^{\nu \lambda
}g^{\beta \gamma }+g^{\mu \gamma }g^{\nu \alpha }g^{\lambda \beta }+g^{\mu
\alpha }g^{\nu \gamma }g^{\lambda \beta }\right)  \notag \\
&&+\Gamma _{\mu \nu \lambda }\Gamma _{\alpha \beta \gamma }\left( g^{\mu
\lambda }g^{\nu \alpha }g^{\beta \gamma }-g^{\mu \gamma }g^{\nu \alpha
}g^{\lambda \beta }\right)  \notag \\
&=&\Gamma _{\mu \nu \lambda }\Gamma _{\alpha \beta \gamma }\left( -g^{\mu
\alpha }g^{\nu \lambda }g^{\beta \gamma }+g^{\mu \alpha }g^{\nu \gamma
}g^{\lambda \beta }\right) ,  \label{R1}
\end{eqnarray}%
where we have employed the identities in (\ref{simpl-1}) and (\ref{simpl-2})
in the first line. In the third line, the first term in the parentheses is
canceling with the fifth term, and the third term in the parentheses is
canceling with the sixth term. We write $R_{2}$ in terms of the metric
tensor 
\begin{eqnarray}
R_{2} &=&g^{\beta \gamma }g^{\alpha \delta }\left( \Gamma _{\delta \beta
\gamma ,\alpha }-\Gamma _{\delta \beta \alpha ,\gamma }\right)  \notag \\
&=&\frac{1}{2}g^{\beta \gamma }g^{\alpha \delta }\left( g_{\beta \delta
,\gamma \alpha }+g_{\gamma \delta ,\beta \alpha }-g_{\beta \gamma ,\delta
\alpha }-g_{\beta \delta ,\alpha \gamma }-g_{\alpha \delta ,\beta \gamma
}+g_{\beta \alpha ,\delta \gamma }\right)  \notag \\
&=&\frac{1}{2}g_{\mu \nu ,\alpha \beta }\left( g^{\mu \alpha }g^{\nu \beta
}+g^{\mu \beta }g^{\nu \alpha }-2g^{\mu \nu }g^{\alpha \beta }\right) .
\label{R2}
\end{eqnarray}%
In the following Lemma, we are stating the conjugate momenta induced by the
Hilbert Lagrangian.

\begin{lemma}
\label{Lemma-5.1}The Ostrogradski's momenta induced by the Hilbert
Lagrangian (\ref{L2}) are 
\begin{eqnarray}
P^{\mu \nu \alpha } &=&\frac{1}{2}\left( -g^{\mu \alpha }\Gamma ^{\nu
}-g^{\nu \alpha }\Gamma ^{\mu }+\Gamma ^{\nu \mu \alpha }+\Gamma ^{\mu \nu
\alpha }\right) \sqrt{-\det g}  \label{P-1-EH-L} \\
P^{\mu \nu \alpha \beta } &=&\frac{1}{2}\left( g^{\mu \alpha }g^{\nu \beta
}+g^{\mu \beta }g^{\nu \alpha }-2g^{\mu \nu }g^{\alpha \beta }\right) \sqrt{%
-\det g},  \label{P-2-EH-L-}
\end{eqnarray}%
where the symbol $\Gamma ^{\nu }$ is the one defined in (\ref{G-T}).
\end{lemma}

\begin{proof}
First recall the definition of the Ostrogradski's momenta%
\begin{equation*}
P^{\mu \nu \alpha }=\frac{\partial L}{\partial g_{\mu \nu ,\alpha }}%
-P_{,\beta }^{\mu \nu \alpha \beta },\text{ \ \ \ }P^{\mu \nu \alpha \beta }=%
\frac{\partial L}{\partial g_{\mu \nu ,\alpha \beta }}.
\end{equation*}%
By substituting the exhibition of the Hilbert Lagrangian given in (\ref{L2}%
), we can rewrite the momenta as%
\begin{equation}
P^{\mu \nu \alpha }=\frac{\partial R_{1}}{\partial g_{\mu \nu ,\alpha }}%
\sqrt{-\det g}-P_{,\beta }^{\mu \nu \alpha \beta },\text{ \ \ \ }P^{\mu \nu
\alpha \beta }=\frac{\partial R_{2}}{\partial g_{\mu \nu ,\alpha \beta }}%
\sqrt{-\det g},  \label{P-1-2}
\end{equation}%
where $R_{1}$ and $R_{2}$ as the ones in (\ref{R1}) and (\ref{R2}),
respectively. It is immediate to observe that the second momenta is 
\begin{equation}
P^{\mu \nu \alpha \beta }=\frac{1}{2}\left( g^{\mu \alpha }g^{\nu \beta
}+g^{\mu \beta }g^{\nu \alpha }-2g^{\mu \nu }g^{\alpha \beta }\right) \sqrt{%
-\det g}.  \label{P-2-EH}
\end{equation}%
Notice from (\ref{P-1-2}) that, in order to determine the first momenta $%
P^{\mu \nu \alpha }$, we need to take the divergence of the second momenta $%
P^{\mu \nu \alpha \beta }$, given in (\ref{P-2-EH}), with respect to $%
x^{\beta }$. For this, we start with taking the partial derivative of $\sqrt{%
-\det g}$ with respect to $x^{\beta }$ as follows%
\begin{eqnarray}
\left( \sqrt{-\det g}\right) _{,\beta } &=&\frac{1}{2\sqrt{-\det g}}(-1)%
\frac{\partial }{\partial x^{\beta }}(\det g)=\frac{-1}{2\sqrt{-\det g}}%
\left( \det g\right) g^{\sigma \tau }g_{\sigma \tau ,\beta }  \notag \\
&=&\frac{1}{2}\sqrt{-\det g}g^{\sigma \tau }g_{\sigma \tau ,\beta }=\frac{1}{%
2}\sqrt{-\det g}g^{\sigma \tau }\left( \Gamma _{\sigma \tau \beta }+\Gamma
_{\tau \sigma \beta }\right)  \notag \\
&=&\sqrt{-\det g}g^{\sigma \tau }\Gamma _{\sigma \tau \beta }=\sqrt{-\det g}%
\Delta _{\beta },  \label{grad-det}
\end{eqnarray}%
where the symbol $\Delta _{\beta }$, in (\ref{D-D}), has been substituted in
the last line of the calculation. On the other hand, we take the divergence 
\begin{eqnarray}
&&\left( g^{\mu \alpha }g^{\nu \beta }+g^{\mu \beta }g^{\nu \alpha }-2g^{\mu
\nu }g^{\alpha \beta }\right) _{,\beta }  \notag \\
\text{ \ \ \ } &=&g_{,\beta }^{\mu \alpha }g^{\nu \beta }+g^{\mu \alpha
}g_{,\beta }^{\nu \beta }+g_{,\beta }^{\mu \beta }g^{\nu \alpha }+g^{\mu
\beta }g_{,\beta }^{\nu \alpha }-2g^{\mu \nu }g_{,\beta }^{\alpha \beta
}-2g_{,\beta }^{\mu \nu }g^{\alpha \beta } \\
&=&-g^{\nu \beta }g_{\beta \rho }\left( \Gamma ^{\mu \alpha \rho }+\Gamma
^{\alpha \mu \rho }\right) -g^{\mu \alpha }g_{\beta \rho }\left( \Gamma
^{\nu \beta \rho }+\Gamma ^{\beta \nu \rho }\right)  \notag \\
&&-g^{\nu \alpha }g_{\beta \rho }\left( \Gamma ^{\mu \beta \rho }+\Gamma
^{\beta \mu \rho }\right) -g^{\mu \beta }g_{\beta \rho }\left( \Gamma ^{\nu
\alpha \rho }+\Gamma ^{\alpha \nu \rho }\right)  \notag \\
&&+2g^{\mu \nu }g_{\beta \rho }\left( \Gamma ^{\alpha \beta \rho }+\Gamma
^{\beta \alpha \rho }\right) +2g^{\alpha \beta }g_{\beta \rho }(\Gamma ^{\mu
\nu \rho }+\Gamma ^{\nu \mu \rho })  \notag \\
&=&-\Gamma ^{\mu \alpha \nu }-\Gamma ^{\alpha \mu \nu }-g^{\mu \alpha
}(\Gamma ^{\nu }+\Delta ^{\nu })-g^{\nu \alpha }(\Gamma ^{\mu }+\Delta ^{\mu
})  \notag \\
&&-\Gamma ^{\nu \alpha \mu }-\Gamma ^{\alpha \nu \mu }+2g^{\mu \nu }\left(
\Gamma ^{\alpha }+\Delta ^{\alpha }\right) +2(\Gamma ^{\mu \nu \alpha
}+\Gamma ^{\nu \mu \alpha })  \notag \\
&=&2g^{\mu \nu }\left( \Gamma ^{\alpha }+\Delta ^{\alpha }\right) -g^{\mu
\alpha }(\Gamma ^{\nu }+\Delta ^{\nu })-g^{\nu \alpha }(\Gamma ^{\mu
}+\Delta ^{\mu })+  \notag \\
&&+2(\Gamma ^{\mu \nu \alpha }+\Gamma ^{\nu \mu \alpha })-(\Gamma ^{\mu
\alpha \nu }+\Gamma ^{\nu \alpha \mu })-(\Gamma ^{\alpha \mu \nu }+\Gamma
^{\alpha \nu \mu }),  \label{div-}
\end{eqnarray}%
where, the identities (\ref{simpl-1}) and (\ref{simpl-2}) have been used in
the second line, and the symbols $\Gamma ^{\alpha }$ and $\Delta ^{\alpha }$%
, defined in (\ref{G-T}) and (\ref{D-T}) have been substituted in the third
line. In the light of the calculations in (\ref{grad-det}) and (\ref{div-}),
the divergence of the second momenta $P^{\mu \nu \alpha \beta }$ turns out
to be 
\begin{eqnarray}
P_{,\beta }^{\mu \nu \alpha \beta } &=&\frac{1}{2}\left( g^{\mu \alpha
}g^{\nu \beta }+g^{\mu \beta }g^{\nu \alpha }-2g^{\mu \nu }g^{\alpha \beta
}\right) \left( \sqrt{-\det g}\right) _{,\beta }  \notag \\
&&+\frac{1}{2}\left( g^{\mu \alpha }g^{\nu \beta }+g^{\mu \beta }g^{\nu
\alpha }-2g^{\mu \nu }g^{\alpha \beta }\right) _{,\beta }\sqrt{-\det g} 
\notag \\
&=&\frac{1}{2}\left( g^{\mu \alpha }g^{\nu \beta }+g^{\mu \beta }g^{\nu
\alpha }-2g^{\mu \nu }g^{\alpha \beta }\right) \sqrt{-\det g}\Delta _{\beta }
\notag \\
&&+\frac{1}{2}(2g^{\mu \nu }\left( \Gamma ^{\alpha }+\Delta ^{\alpha
}\right) -g^{\mu \alpha }(\Gamma ^{\nu }+\Delta ^{\nu })-g^{\nu \alpha
}(\Gamma ^{\mu }+\Delta ^{\mu }))\sqrt{-\det g}  \notag \\
&&+\frac{1}{2}(\Gamma ^{\mu \nu \alpha }+\Gamma ^{\nu \mu \alpha }-2\Gamma
^{\alpha \mu \nu })\sqrt{-\det g}  \notag \\
&=&\frac{1}{2}\left( 2g^{\mu \nu }\Gamma ^{\alpha }-g^{\mu \alpha }\Gamma
^{\nu }-g^{\nu \alpha }\Gamma ^{\mu }+\Gamma ^{\mu \nu \alpha }+\Gamma ^{\nu
\mu \alpha }-2\Gamma ^{\alpha \mu \nu }\right) \sqrt{-\det g},  \label{divP}
\end{eqnarray}%
where the identity $g^{\nu \beta }\Delta _{\beta }=\Delta ^{\nu }$ has been
used. Notice that all the terms involving $\Delta ^{\nu }$ canceling each
other in the calculation. Let us now concentrate on the first term $\partial
R_{1}/\partial g_{\mu \nu ,\alpha }\sqrt{-\det g}$ in the momenta $P^{\mu
\nu \alpha }$, applying the chain rule, we have that%
\begin{equation*}
\frac{\partial R_{1}}{\partial g_{\alpha \beta ,\gamma }}\sqrt{-\det g}=%
\frac{\partial R_{1}}{\partial \Gamma _{\lambda \mu \nu }}\frac{\partial
\Gamma _{\lambda \mu \nu }}{\partial g_{\alpha \beta ,\gamma }}\sqrt{-\det g}%
.
\end{equation*}%
Notice that, the partial derivative of $R_{1}$ with respect to the
Christoffel symbol of the first kind $\Gamma _{\lambda \mu \nu }$ is
computed to be%
\begin{equation}
\frac{\partial R_{1}}{\partial \Gamma _{\lambda \mu \nu }}=2\Gamma _{\alpha
\beta \gamma }\left( -g^{\lambda \alpha }g^{\mu \nu }g^{\beta \gamma
}+g^{\lambda \alpha }g^{\mu \gamma }g^{\nu \beta }\right) ,  \label{R-1-der}
\end{equation}%
whereas the partial derivative of $\Gamma _{\lambda \mu \nu }$ with respect
to $g_{\alpha \beta ,\gamma }$ is 
\begin{equation}
\frac{\partial \Gamma _{\lambda \mu \nu }}{\partial g_{\alpha \beta ,\gamma }%
}=\frac{1}{4}\left( \left( \delta _{\mu }^{\alpha }\delta _{\lambda }^{\beta
}\delta _{\nu }^{\gamma }+\delta _{\nu }^{\alpha }\delta _{\lambda }^{\beta
}\delta _{\mu }^{\gamma }-\delta _{\mu }^{\alpha }\delta _{\nu }^{\beta
}\delta _{\lambda }^{\gamma }\right) +\left( \delta _{\mu }^{\beta }\delta
_{\lambda }^{\alpha }\delta _{\nu }^{\gamma }+\delta _{\nu }^{\beta }\delta
_{\lambda }^{\alpha }\delta _{\mu }^{\gamma }-\delta _{\mu }^{\beta }\delta
_{\nu }^{\alpha }\delta _{\lambda }^{\gamma }\right) \right) .  \label{G-G}
\end{equation}%
Here, the factor $1/4$ is the manifestation of the symmetry of the metric
tensor. We multiply the expressions (\ref{R-1-der}) and (\ref{G-G}) and
arrange the terms, so that we arrive at%
\begin{equation}
\frac{\partial R_{2}}{\partial g_{\mu \nu ,\alpha }}\sqrt{-\det g}=g^{\mu
\nu }\Gamma ^{\alpha }-g^{\mu \alpha }\Gamma ^{\nu }-g^{\nu \alpha }\Gamma
^{\mu }-\Gamma ^{\alpha \mu \nu }+\Gamma ^{\nu \mu \alpha }+\Gamma ^{\mu \nu
\alpha }\sqrt{-\det g}.  \label{R-G}
\end{equation}%
Now we are ready to write the first momenta $P^{\mu \nu \alpha }$, for this
simply take the difference of (\ref{R-G}) and (\ref{divP}), this gives 
\begin{eqnarray}
P^{\mu \nu \alpha } &=&\frac{\partial R_{2}}{\partial g_{\mu \nu ,\alpha }}%
\sqrt{-\det g}-P_{,\beta }^{\mu \nu \alpha \beta }  \notag \\
&=&(g^{\mu \nu }\Gamma ^{\alpha }-g^{\mu \alpha }\Gamma ^{\nu }-g^{\nu
\alpha }\Gamma ^{\mu }-\Gamma ^{\alpha \mu \nu }+\Gamma ^{\nu \mu \alpha
}+\Gamma ^{\mu \nu \alpha })\sqrt{-\det g}  \notag \\
&&-\frac{1}{2}\left( 2g^{\mu \nu }\Gamma ^{\alpha }-g^{\mu \alpha }\Gamma
^{\nu }-g^{\nu \alpha }\Gamma ^{\mu }-2\Gamma ^{\alpha \mu \nu }+\Gamma
^{\mu \nu \alpha }+\Gamma ^{\nu \mu \alpha }\right) \sqrt{-\det g}  \notag \\
&=&\frac{1}{2}\left( -g^{\mu \alpha }\Gamma ^{\nu }-g^{\nu \alpha }\Gamma
^{\mu }+\Gamma ^{\nu \mu \alpha }+\Gamma ^{\mu \nu \alpha }\right) \sqrt{%
-\det g}  \label{P-1-EH}
\end{eqnarray}%
where the first and fourth terms in the second and the third lines are
canceling each other, respectively.
\end{proof}

We are now ready to prove the Proposition $5$. In the present framework, the
De Donder form turns out to be 
\begin{equation*}
\Theta _{\mathrm{Hilbert}}=R_{1}\sqrt{-\det g}\mathrm{d}_{4}x+R_{2}\sqrt{%
-\det g}\mathrm{d}_{4}x+\Xi _{\mathrm{Hilbert}},
\end{equation*}%
where $\Xi _{\mathrm{Hilbert}}$ is the boundary form induced by the Hilbert
Lagrangian. Explicitly, the boundary form is 
\begin{equation}
\Xi _{\mathrm{Hilbert}}=(P^{\alpha \beta \mu }\mathrm{d}g_{\alpha \beta
}+P^{\alpha \beta \mu \nu }\mathrm{d}g_{\alpha \beta ,\nu })\wedge (\frac{%
{\small \partial }}{{\small \partial x}^{\mu }}%
{\mbox{$ \rule {5pt} {.5pt}\rule {.5pt} {6pt} \, $}}%
\mathrm{d}_{4}x)-(P^{\alpha \beta \mu }g_{\alpha \beta ,\mu }+P^{\alpha
\beta \mu \nu }g_{\alpha \beta ,\nu \mu })\mathrm{d}_{4}x.  \notag
\end{equation}%
By substituting the conjugate momenta $P^{\mu \nu \alpha }$ and $P^{\mu \nu
\alpha \beta }$, respectively given in (\ref{P-1-EH-L}) and (\ref{P-2-EH-L-}%
), one has 
\begin{eqnarray}
\Xi _{\mathrm{Hilbert}} &=&\frac{1}{2}\left( -g^{\alpha \mu }\Gamma ^{\beta
}-g^{\beta \mu }\Gamma ^{\alpha }+\Gamma ^{\beta \alpha \mu }+\Gamma
^{\alpha \beta \mu }\right) \sqrt{-\det g}\mathrm{d}g_{\alpha \beta }\wedge (%
\frac{{\small \partial }}{{\small \partial x}^{\mu }}%
{\mbox{$ \rule {5pt} {.5pt}\rule {.5pt} {6pt} \, $}}%
\mathrm{d}_{4}x)  \notag \\
&&+\frac{1}{2}\left( g^{\alpha \mu }g^{\beta \nu }+g^{\alpha \nu }g^{\beta
\mu }-2g^{\alpha \beta }g^{\mu \nu }\right) \sqrt{-\det g}\mathrm{d}%
g_{\alpha \beta ,\nu }\wedge (\frac{{\small \partial }}{{\small \partial x}%
^{\mu }}%
{\mbox{$ \rule {5pt} {.5pt}\rule {.5pt} {6pt} \, $}}%
\mathrm{d}_{4}x)  \notag \\
&&-\frac{1}{2}\left( -g^{\alpha \mu }\Gamma ^{\beta }-g^{\beta \mu }\Gamma
^{\alpha }+\Gamma ^{\beta \alpha \mu }+\Gamma ^{\alpha \beta \mu }\right) 
\sqrt{-\det g}g_{\alpha \beta ,\mu }\mathrm{d}_{4}x  \notag \\
&&-\frac{1}{2}\left( g^{\alpha \mu }g^{\beta \nu }+g^{\alpha \nu }g^{\beta
\mu }-2g^{\alpha \beta }g^{\mu \nu }\right) \sqrt{-\det g}g_{\alpha \beta
,\nu \mu }\mathrm{d}_{4}x.
\end{eqnarray}%
Substitution of the boundary form $\Xi _{\mathrm{Hilbert}}$ and the terms $%
R_{1}$ and $R_{2}$ in (\ref{R1}) and (\ref{R2}) leads to the following
expression of the De Donder form 
\begin{eqnarray*}
\Theta _{\mathrm{Hilbert}} &=&\frac{1}{2}g_{\mu \nu ,\alpha \beta }\left(
g^{\mu \alpha }g^{\nu \beta }+g^{\mu \beta }g^{\nu \alpha }-2g^{\mu \nu
}g^{\alpha \beta }\right) \sqrt{-\det g}\mathrm{d}_{4}x \\
&&+\Gamma _{\lambda \mu \nu }\Gamma _{\alpha \beta \gamma }\left(
-g^{\lambda \alpha }g^{\mu \nu }g^{\beta \gamma }+g^{\lambda \alpha }g^{\mu
\gamma }g^{\nu \beta }\right) \sqrt{-\det g}\mathrm{d}_{4}x \\
&&+\frac{1}{2}\left( -g^{\alpha \mu }\Gamma ^{\beta }-g^{\beta \mu }\Gamma
^{\alpha }+\Gamma ^{\beta \alpha \mu }+\Gamma ^{\alpha \beta \mu }\right) 
\sqrt{-\det g}\mathrm{d}g_{\alpha \beta }\wedge (\frac{{\small \partial }}{%
{\small \partial x}^{\mu }}%
{\mbox{$ \rule {5pt} {.5pt}\rule {.5pt} {6pt} \, $}}%
\mathrm{d}_{4}x) \\
&&+\frac{1}{2}\left( g^{\alpha \mu }g^{\beta \nu }+g^{\alpha \nu }g^{\beta
\mu }-2g^{\alpha \beta }g^{\mu \nu }\right) \sqrt{-\det g}\mathrm{d}%
g_{\alpha \beta ,\nu }\wedge (\frac{{\small \partial }}{{\small \partial x}%
^{\mu }}%
{\mbox{$ \rule {5pt} {.5pt}\rule {.5pt} {6pt} \, $}}%
\mathrm{d}_{4}x) \\
&&-\frac{1}{2}\left( -g^{\alpha \mu }\Gamma ^{\beta }-g^{\beta \mu }\Gamma
^{\alpha }+\Gamma ^{\beta \alpha \mu }+\Gamma ^{\alpha \beta \mu }\right)
g_{\alpha \beta ,\mu }\mathrm{d}_{4}x \\
&&-\frac{1}{2}g_{\alpha \beta ,\nu \mu }\left( g^{\alpha \mu }g^{\beta \nu
}+g^{\alpha \nu }g^{\beta \mu }-2g^{\alpha \beta }g^{\mu \nu }\right) \sqrt{%
-\det g}\mathrm{d}_{4}x.
\end{eqnarray*}%
Notice that, the first and the last terms are canceling since they are minus
of the each other. So that there remain 
\begin{eqnarray*}
\Theta _{\mathrm{Hilbert}} &=&\left( \Gamma _{\lambda \mu \nu }\Gamma
_{\alpha \beta \gamma }\left( -g^{\lambda \alpha }g^{\mu \nu }g^{\beta
\gamma }+g^{\lambda \alpha }g^{\mu \gamma }g^{\nu \beta }\right) \right) 
\sqrt{-\det g}\mathrm{d}_{4}x \\
&&-\frac{1}{2}\left( -g^{\alpha \mu }\Gamma ^{\beta }-g^{\beta \mu }\Gamma
^{\alpha }+\Gamma ^{\beta \alpha \mu }+\Gamma ^{\alpha \beta \mu }\right)
g_{\alpha \beta ,\mu }\sqrt{-\det g}\mathrm{d}_{4}x \\
&&+\frac{1}{2}\left( -g^{\alpha \mu }\Gamma ^{\beta }-g^{\beta \mu }\Gamma
^{\alpha }+\Gamma ^{\beta \alpha \mu }+\Gamma ^{\alpha \beta \mu }\right) 
\sqrt{-\det g}\mathrm{d}g_{\alpha \beta }\wedge (\frac{{\small \partial }}{%
{\small \partial x}^{\mu }}%
{\mbox{$ \rule {5pt} {.5pt}\rule {.5pt} {6pt} \, $}}%
\mathrm{d}_{4}x) \\
&&+\frac{1}{2}\left( g^{\alpha \mu }g^{\beta \nu }+g^{\alpha \nu }g^{\beta
\mu }-2g^{\alpha \beta }g^{\mu \nu }\right) \sqrt{-\det g}\mathrm{d}%
g_{\alpha \beta ,\nu }\wedge (\frac{{\small \partial }}{{\small \partial x}%
^{\mu }}%
{\mbox{$ \rule {5pt} {.5pt}\rule {.5pt} {6pt} \, $}}%
\mathrm{d}_{4}x).
\end{eqnarray*}%
Let us simplify concentrate on the second line of this expression. A simple
calculations give 
\begin{eqnarray*}
&&\frac{1}{2}\left( -g^{\alpha \mu }\Gamma ^{\beta }-g^{\beta \mu }\Gamma
^{\alpha }+\Gamma ^{\beta \alpha \mu }+\Gamma ^{\alpha \beta \mu }\right)
g_{\alpha \beta ,\mu } \\
&=&\frac{1}{2}\left( -g^{\alpha \mu }\Gamma ^{\beta }-g^{\beta \mu }\Gamma
^{\alpha }+\Gamma ^{\beta \alpha \mu }+\Gamma ^{\alpha \beta \mu }\right)
\left( \Gamma _{\alpha \beta \mu }+\Gamma _{\beta \alpha \mu }\right) \\
&=&-\frac{1}{2}g^{\alpha \mu }g^{uw}g^{\beta r}\Gamma _{ruw}\Gamma _{\alpha
\beta \mu }-\frac{1}{2}g^{\alpha \mu }g^{uw}g^{\beta r}\Gamma _{ruw}\Gamma
_{\beta \alpha \mu }-\frac{1}{2}g^{\beta \mu }g^{uw}g^{\alpha r}\Gamma
_{ruw}\Gamma _{\alpha \beta \mu } \\
&&-\frac{1}{2}g^{\beta \mu }g^{uw}g^{\alpha r}\Gamma _{ruw}\Gamma _{\beta
\alpha \mu }+\frac{1}{2}g^{\beta a}g^{\alpha b}g^{\mu c}\Gamma _{abc}\Gamma
_{\alpha \beta \mu }+\frac{1}{2}g^{\beta a}g^{\alpha b}g^{\mu c}\Gamma
_{abc}\Gamma _{\beta \alpha \mu } \\
&&+\frac{1}{2}g^{\alpha a}g^{\beta b}g^{\mu c}\Gamma _{abc}\Gamma _{\alpha
\beta \mu }+\frac{1}{2}g^{\alpha a}g^{\beta b}g^{\mu c}\Gamma _{abc}\Gamma
_{\beta \alpha \mu } \\
&=&\frac{1}{2}\Gamma _{\lambda \mu \nu }\Gamma _{\alpha \beta \gamma
}(-g^{\alpha \gamma }g^{\mu \nu }g^{\beta \lambda }-g^{\beta \gamma }g^{\mu
\nu }g^{\alpha \lambda }-g^{\beta \gamma }g^{\mu \nu }g^{\alpha \lambda
}-g^{\alpha \gamma }g^{\mu \nu }g^{\beta \lambda }) \\
&&+\frac{1}{2}\Gamma _{\lambda \mu \nu }\Gamma _{\alpha \beta \gamma }\left(
g^{\beta \lambda }g^{\alpha \mu }g^{\gamma \nu }+g^{\alpha \lambda }g^{\beta
\mu }g^{\gamma \nu }+g^{\alpha \lambda }g^{\beta \mu }g^{\gamma \nu
}+g^{\beta \lambda }g^{\alpha \mu }g^{\gamma \nu }\right) \\
&=&\Gamma _{\lambda \mu \nu }\Gamma _{\alpha \beta \gamma }(g^{\beta \lambda
}g^{\alpha \mu }g^{\gamma \nu }+g^{\alpha \lambda }g^{\beta \mu }g^{\gamma
\nu }-g^{\alpha \gamma }g^{\mu \nu }g^{\beta \lambda }-g^{\beta \gamma
}g^{\mu \nu }g^{\alpha \lambda }),
\end{eqnarray*}%
where we used the identity (\ref{simpl-1}) in the first line, and the
identity (\ref{Gamma-contravariant}) in the third line, and in the fourth
line we sum up the similar terms. This simplification reads that the
coefficient of the basis $\sqrt{-\det g}\mathrm{d}_{4}x$ can be written as 
\begin{align*}
&\Gamma _{\lambda \mu \nu }\Gamma _{\alpha \beta \gamma }\left( -g^{\lambda
\alpha }g^{\mu \nu }g^{\beta \gamma }+g^{\lambda \alpha }g^{\mu \gamma
}g^{\nu \beta }\right) \\
&-\Gamma _{\lambda \mu \nu }\Gamma _{\alpha \beta \gamma }(g^{\beta \lambda
}g^{\alpha \mu }g^{\gamma \nu }+g^{\alpha \lambda }g^{\beta \mu }g^{\gamma
\nu }-g^{\alpha \gamma }g^{\mu \nu }g^{\beta \lambda }-g^{\beta \gamma
}g^{\mu \nu }g^{\alpha \lambda }) \\
&=\Gamma _{\lambda \mu \nu }\Gamma _{\alpha \beta \gamma }g^{\lambda \alpha
}g^{\mu \gamma }g^{\nu \beta }+\Gamma _{\lambda \mu \nu }\Gamma _{\alpha
\beta \gamma }g^{\alpha \gamma }g^{\mu \nu }g^{\beta \lambda }-\Gamma
_{\lambda \mu \nu }\Gamma _{\alpha \beta \gamma }g^{\beta \lambda }g^{\alpha
\mu }g^{\gamma \nu }-\Gamma _{\lambda \mu \nu }\Gamma _{\alpha \beta \gamma
}g^{\alpha \lambda }g^{\beta \mu }g^{\gamma \nu } \\
&=\Gamma _{\lambda \mu \nu }\Gamma _{\alpha \beta \gamma }g^{\lambda \alpha
}g^{\mu \gamma }g^{\nu \beta }+\Gamma _{\lambda \mu \nu }\Gamma _{\alpha
\beta \gamma }g^{\alpha \gamma }g^{\mu \nu }g^{\beta \lambda }-\Gamma
_{\lambda \mu \nu }\Gamma _{\alpha \beta \gamma }g^{\beta \lambda }g^{\alpha
\mu }g^{\gamma \nu }-\Gamma _{\lambda \mu \nu }\Gamma _{\alpha \beta \gamma
}g^{\alpha \lambda }g^{\beta \nu }g^{\gamma \mu } \\
&=\Gamma _{\lambda \mu \nu }\Gamma _{\alpha \beta \gamma }g^{\alpha \gamma
}g^{\mu \nu }g^{\beta \lambda }-\Gamma _{\lambda \mu \nu }\Gamma _{\alpha
\beta \gamma }g^{\beta \lambda }g^{\alpha \mu }g^{\gamma \nu },
\end{align*}%
where we have canceled the first and last terms in the first line, and used
the symmetry of the Christoffel symbol in the third line. Eventually, the De
Donder form for Hilbert Lagrangian becomes 
\begin{eqnarray*}
\Theta _{\mathrm{Hilbert}} &=&\Gamma _{\lambda \mu \nu }\Gamma _{\alpha
\beta \gamma }\left( g^{\alpha \gamma }g^{\mu \nu }g^{\beta \lambda
}-g^{\beta \lambda }g^{\alpha \mu }g^{\gamma \nu }\right) \sqrt{-\det g}%
\mathrm{d}_{4}x \\
&&+\frac{1}{2}\left( -g^{\alpha \mu }\Gamma ^{\beta }-g^{\beta \mu }\Gamma
^{\alpha }+\Gamma ^{\beta \alpha \mu }+\Gamma ^{\alpha \beta \mu }\right) 
\sqrt{-\det g}\mathrm{d}g_{\alpha \beta }\wedge (\frac{{\small \partial }}{%
{\small \partial x}^{\mu }}%
{\mbox{$ \rule {5pt} {.5pt}\rule {.5pt} {6pt} \, $}}%
\mathrm{d}_{4}x) \\
&&+\frac{1}{2}\left( g^{\alpha \mu }g^{\beta \nu }+g^{\alpha \nu }g^{\beta
\mu }-2g^{\alpha \beta }g^{\mu \nu }\right) \sqrt{-\det g}\mathrm{d}%
g_{\alpha \beta ,\nu }\wedge (\frac{{\small \partial }}{{\small \partial x}%
^{\mu }}%
{\mbox{$ \rule {5pt} {.5pt}\rule {.5pt} {6pt} \, $}}%
\mathrm{d}_{4}x).
\end{eqnarray*}%
A comparison of this form and the one in (\ref{3.9b}) shows that the proof
of the Proposition 5 in Section 3.2 is achieved.

\subsection{Proof of Proposition 6}

Referring to the Proposition \ref{Proposition 3.3}, to prove the Proposition %
\ref{Proposition 3.2} we need to just focus on the De Donder form (\ref{3.12}%
) for the matter. See that it is composed of a Lagrangian term and the
boundary term. We label the boundary term as 
\begin{equation*}
\Xi _{3}=q^{\mu }(\mathrm{d}t-z_{\nu }\mathrm{d}x^{\nu })\wedge (\partial
_{\mu }%
{\mbox{$ \rule {5pt} {.5pt}\rule {.5pt} {6pt} \, $}}%
\mathrm{d}_{4}x),
\end{equation*}%
where the coefficient function $q^{\mu }$ reads (\ref{gM}) for a local
section. Let us first show that, $\Xi _{3}$ is invariant under a coordinate
transformation on the base manifold $M$ given in (\ref{gM}). See that 
\begin{eqnarray*}
q^{\mu } &=&g^{\mu \nu }\phi _{,\nu }\sqrt{-\det \left( g_{\mu \nu }\right) }%
=g^{\mu ^{\prime }\nu ^{\prime }}x_{,\mu ^{\prime }}^{\mu }x_{,\nu ^{\prime
}}^{\nu }\phi _{,\alpha ^{\prime }}x_{,\nu }^{\alpha ^{\prime }}\sqrt{-\det
(g_{\mu ^{\prime }\nu ^{\prime }}x_{,\mu }^{\mu ^{\prime }}x_{,\nu }^{\nu
^{\prime }})} \\
&=&\det (x_{,\lambda }^{\lambda ^{\prime }})x_{,\mu ^{\prime }}^{\mu }g^{\mu
^{\prime }\nu ^{\prime }}\phi _{,\nu ^{\prime }}\sqrt{-\det (g_{\mu ^{\prime
}\nu ^{\prime }})} \\
&=&q^{\mu ^{\prime }}\det (x_{,\lambda }^{\lambda ^{\prime }})x_{,\mu
^{\prime }}^{\mu }.
\end{eqnarray*}%
For the basis we recall the transformation in (\ref{Wedge-trf}), and compute 
\begin{equation*}
(\mathrm{d}t-z_{\nu }\mathrm{d}x^{\nu })=(\mathrm{d}t-z_{\nu ^{\prime
}}x_{,\nu }^{\nu ^{\prime }}x_{,\alpha }^{\nu }\mathrm{d}x^{\alpha ^{\prime
}})=\mathrm{d}t-z_{\nu ^{\prime }}\mathrm{d}x^{\nu ^{\prime }}.
\end{equation*}%
Collecting all these, one sees that formulation of the boundary form remains
the same under coordinate transformation%
\begin{eqnarray*}
&&q^{\mu }\sqrt{-\det g}(\mathrm{d}t-z_{\nu }\mathrm{d}x^{\nu })\wedge
(\partial _{\mu }%
{\mbox{$ \rule {5pt} {.5pt}\rule {.5pt} {6pt} \, $}}%
\mathrm{d}_{4}x) \\
&=&q^{\mu ^{\prime }}\det (x_{,\lambda }^{\lambda ^{\prime }})x_{,\mu
^{\prime }}^{\mu }(\mathrm{d}t-z_{\nu ^{\prime }}\mathrm{d}x^{\nu ^{\prime
}})\wedge (\det (x_{,\lambda ^{\prime }}^{\lambda })x_{,\mu }^{\alpha
^{\prime }}\partial _{\alpha ^{\prime }}%
{\mbox{$ \rule {5pt} {.5pt}\rule {.5pt} {6pt} \, $}}%
\mathrm{d}_{4}x^{\prime }) \\
&=&q^{\mu ^{\prime }}x_{,\mu ^{\prime }}^{\mu }x_{,\mu }^{\alpha ^{\prime }}%
\sqrt{-\det (g^{\prime })}(\mathrm{d}t-z_{\nu ^{\prime }}\mathrm{d}x^{\nu
^{\prime }})\wedge (\partial _{\alpha ^{\prime }}%
{\mbox{$ \rule {5pt} {.5pt}\rule {.5pt} {6pt} \, $}}%
\mathrm{d}_{4}x^{\prime }) \\
&=&q^{\mu ^{\prime }}(\mathrm{d}t-z_{\nu ^{\prime }}\mathrm{d}x^{\nu
^{\prime }})\wedge (\partial _{\mu ^{\prime }}%
{\mbox{$ \rule {5pt} {.5pt}\rule {.5pt} {6pt} \, $}}%
\mathrm{d}_{4}x^{\prime }).
\end{eqnarray*}%
On the other hand, the Lagrangian term is%
\begin{eqnarray*}
L_{\mathrm{matter}}^{g}\mathrm{d}_{4}x &=&(\frac{{\small 1}}{{\small 2}}%
g^{\mu \nu }z_{\mu }z_{\nu }+V(z))\sqrt{-\det g}\mathrm{d}_{4}x \\
&=&(\frac{{\small 1}}{{\small 2}}g^{\mu ^{\prime }\nu ^{\prime }}x_{,\mu
^{\prime }}^{\mu }x_{,\nu ^{\prime }}^{\nu }z_{\alpha ^{\prime }}z_{\beta
^{\prime }}x_{,\mu }^{\alpha ^{\prime }}x_{,\nu }^{\beta ^{\prime
}}+V(z^{\prime }))\sqrt{-\det g^{\prime }}\mathrm{d}_{4}x^{\prime } \\
&=&(\frac{{\small 1}}{{\small 2}}g^{\mu ^{\prime }\nu ^{\prime }}z_{\mu
^{\prime }}z_{\nu ^{\prime }}+V(z^{\prime }))\sqrt{-\det g^{\prime }}\mathrm{%
d}_{4}x^{\prime },
\end{eqnarray*}%
where we have employed the fact that $\sqrt{-\det g}\mathrm{d}_{4}x$ is
invariant under the coordinate transformation. Notice that, assumption of
the invariance of the function $V$ leads to the invariance of $L_{\mathrm{%
matter}}^{g}\mathrm{d}_{4}x$.\smallskip

\section{Acknowledgement}

The authors are grateful to Department of Mathematics and Statistics,
University of Victoria, for hosting a Workshop on Geometry and Mechanics, 16
-20 July, 2018, where they began collaboration on this paper.

The visit of (OE) was supported by T\"{U}B\.{I}TAK (the Scientific and
Technological Research Council of Turkey) under the project title "Matched
pairs of Lagrangian and Hamiltonian Systems" with the project number 117F426.

\end{document}